\documentclass[hidelinks,onefignum,onetabnum]{siamart251216}

\usepackage{amssymb}
\usepackage{mathtools}
\usepackage{bbm}
\usepackage{algorithmic}
\usepackage{booktabs}
\graphicspath{{figs/}}
\hypersetup{
    colorlinks=true,
    linkcolor=blue,
    anchorcolor=blue,
    citecolor=blue,
    urlcolor=blue
}

\newsiamremark{remark}{Remark}
\crefname{remark}{Remark}{Remarks}
\Crefname{remark}{Remark}{Remarks}
\usepackage{etoolbox}
\AtBeginEnvironment{remark}{\crefalias{theorem}{remark}}

\newcommand{\R}{\mathbb{R}}
\newcommand{\C}{\mathbb{C}}

\newcommand{\Id}{\mathbbm{1}}                 
\newcommand{\qhat}{\hat{q}}                   
\newcommand{\phat}{\hat{p}}                   
\newcommand{\Order}{\mathcal{O}}
\newcommand{\Htot}{H_{\mathrm{total}}}
\newcommand{\tr}{\operatorname{tr}}

\title{A Provable Oracle-Free Quantum Algorithm for Nonlinear Dynamics on
       Hybrid Oscillator--Qubit Processors
       \thanks{%
         Submitted to the editors \today.}}

\author{%
    Kausthubh Chandramouli
    \thanks{%
    Department of Electrical and Computer Engineering, North Carolina State University, Raleigh, NC 27606, USA
    (\email{kprabha@ncsu.edu}, \email{yliu335@ncsu.edu}).}%
    \and
    Yan Li
    \thanks{%
    Department of Electrical Engineering, Pennsylvania State University, University Park, PA 16802, USA
    (\email{yql5925@psu.edu}).}%
    \and
    Yuan Liu\footnotemark[2]
    \thanks{Department of Computer Science, North Carolina State University, Raleigh, North Carolina 27695, USA.
    Also affiliated with: Department of Physics and Astronomy, North Carolina State University, Raleigh, North Carolina 27695, USA}
}

\headers{Provable Oracle-Free Quantum Algorithm for Nonlinear Dynamics}{K.\ Chandramouli, Y. Li, and Y. Liu}

\ifpdf
\hypersetup{
  pdftitle={A Provable, Oracle-free Quantum Algorithm for Nonlinear Dynamics
            on Hybrid Oscillator-Qubit Processors},
  pdfauthor={K. Chandramouli}
}
\fi

\begin{document}
\maketitle

\begin{abstract}
  We develop a hybrid qubit--qumode algorithm for nonlinear ordinary
  differential equations $\dot{\mathbf{x}}=\mathbf{f}(\mathbf{x})$ with a
  drift of polynomial degree~$L$.  The algorithm propagates the
  Fokker--Planck density rather than the state, returning the trajectory as
  the density peak in the small-noise limit together with the low-order
  statistics of the flow.  Schr\"{o}dingerisation carries the discretised
  generator into a parametrised family of Schr\"{o}dinger equations, whose
  Fourier-mode parameter we place on a single continuous-variable qumode, so
  the whole continuum evolves under one coherent circuit.  A tridiagonal generator is known to admit a
  bipartite Pauli decomposition, $\{I,Z\}$ on a prefix register and
  $\{X,Y\}$ on a suffix register, sorting its strings into $\Order(\log N)$
  commuting families [Arseniev et al., 2024].  Our contribution is what a
  polynomial drift adds to that structure: each family factorises into a
  diagonal of degree at most $L$ tensored with a fixed rank-two bond
  operator, and carries only $\Order(2^{m}n^{L})$ non-zero strings at carry
  level $m$.  The total string count is $\Theta(N)$; the saving comes from
  the factorisation, which makes every family exponential an exact product
  of $\Order(n^{L})$ monomial-controlled momentum displacements, with no
  intra-family Trotter error.  The circuit costs $\Order(d^{L+1}n^{L+2})$ gates per Trotter step
  on a $d$-dimensional grid of $N=2^{n}$ points per axis, polylogarithmic in
  the $N^{d}$ grid points, and invokes no sparse-access oracle or block
  encoding: every gate is fixed in closed form by the drift coefficients.
  We also bound the numerical abscissa $\lambda_{\max}(H_{1})$ uniformly in
  the grid spacing, which fixes the recovery domain of the warped-phase
  transform and the post-selection cost, and we validate the whole pipeline
  classically on two nonlinear benchmarks.
\end{abstract}



\section{Introduction}\label{sec:intro}

Many problems in computational science, among them chemical kinetics, fluid
dynamics, plasma physics, power systems, and neural dynamics, are governed by
nonlinear ordinary differential equations
\begin{equation}
  \dot{\mathbf{x}} \;=\; \mathbf{f}(\mathbf{x}),\qquad
  \mathbf{x}\in\R^{d},
  \label{eq:ode_intro}
\end{equation}
whose flow is the object we compute: the trajectory issuing from a given
initial state, and the low-order statistics of the bundle issuing from a
distribution of them.  We reach it, as Tennie and Magri
do~\cite{TennieMagri2024,TennieMagri2025}, by propagating the state
\emph{density} rather than the state.  A small stochastic forcing makes the
density obey the Fokker--Planck equation, linear in the density however
nonlinear $\mathbf{f}$ is in $\mathbf{x}$, and the deterministic trajectory
is returned as its peak in the small-noise limit.  The price is dimension: a
grid density on $\R^{d}$ carries $\Order(N^{d})$ degrees of freedom, so
removing the nonlinearity buys an exponentially large state space in
exchange.  Quantum computing is the natural response, that state space being
addressed by $d\log_{2}N$ qubits, with near-optimal gate complexities for
Hamiltonian simulation of sparse linear
systems~\cite{Berry2015,Low2019,Childs2021}; efficient quantum algorithms
are, however, available primarily for \emph{linear} systems, and the
Fokker--Planck embedding supplies the linear form they require.

The construction has four stages.  \emph{(i)~Fokker--Planck linearisation:}
forcing \eqref{eq:ode_intro} with a small Wiener term, the density of the
perturbed state evolves linearly~\cite{Risken1989,TennieMagri2024}.
\emph{(ii)~Spatial discretisation:} central differences on a uniform grid
give $\dot{\boldsymbol{\rho}}=A\boldsymbol{\rho}$ with $A$ sparse,
tridiagonal, and non-Hermitian.  \emph{(iii)~Schr\"{o}dingerisation:} the
warped phase transformation~\cite{JinLiu2022,JinLiu2024} maps $e^{At}$ to
the family $i\,\partial_{t}\tilde{w}=(2\pi\eta H_{1}+H_{2})\tilde{w}$,
$\eta\in\R$, with $H_{1}=(A+A^{\dagger})/2$ and
$H_{2}=(A^{\dagger}-A)/(2i)$.  \emph{(iv)~Continuous-variable LCU:} rather
than discretise $\eta$ and simulate each mode, we couple the qubit register
carrying $\boldsymbol{\rho}$ to a \emph{qumode} whose position quadrature
$\qhat$ plays the role of $\eta$; initialised in the Lorentzian state and
evolved under
\begin{equation}
  \Htot \;=\; H_{1}\otimes\qhat + H_{2}\otimes\Id_{\mathrm{qumode}},
  \label{eq:Htot_intro}
\end{equation}
it realises the whole $\eta$-continuum in one coherent circuit.  Stages~(i)
and~(ii), and the semi-stability underlying~(iii), are due to Tennie and
Magri~\cite{TennieMagri2024,TennieMagri2025}, who price the mode continuum on
a Fourier register but leave open the circuit: how the exponentials
$e^{-i\mathcal{H}(\eta)\Delta t}$ are compiled, at what gate cost, and from
what input data.  This paper answers those three questions.

Two things make that possible.  The first is structural and is not ours:
Arseniev et al.~\cite{Arseniev2024} show that a tridiagonal matrix on
$N=2^{n}$ points has Pauli support confined to $n+1$ bipartite sets
$\{I,Z\}^{\otimes(n-m)}\otimes\{X,Y\}^{\otimes m}$, and that a real
tridiagonal matrix therefore populates at most $2n+1$ internally commuting
families.  We take that as given and ask what a \emph{polynomial} drift adds
to it, which is the second: each family then factorises into a diagonal of
degree at most $L$ tensored with a fixed rank-two bond operator, carries only
$\Order(2^{m}n^{L})$ non-zero strings at carry level $m$, and has its
exponential realised \emph{exactly}, with no intra-family Trotter error, as a
product of monomial-controlled qumode displacements.  The circuit costs
$\Order(d^{L+1}n^{L+2})$ gates per Trotter step, polynomial in the dimension
and polylogarithmic in the $N^{d}$ grid points.  We also bound the numerical
abscissa $\lambda_{\max}(H_{1})$ uniformly in the grid spacing, which fixes
the shifted half-line from which the density is recovered and the
post-selection cost it carries.

We call the algorithm \emph{oracle-free} in a definite sense.  Quantum
differential-equation solvers built on quantum linear-systems
routines~\cite{Harrow2009,Childs2021,Krovi2023} are normally stated relative
to a black-box oracle returning the non-zero entries of the coefficient
matrix, or to a block encoding of it, with the cost of realising that
primitive left outside the analysis.  No such primitive appears here: the
entire input to the compiler is the coefficients of the degree-$L$ drift,
the grid parameters, and the diffusion, from which a closed-form
preprocessing step fixes every gate.  At no point is $A$ assembled, queried,
or accessed through an intermediary, so every count we quote is a count of
Cliffords, phase rotations, and controlled qumode displacements, the last
counted as elementary and native on the platforms concerned.  The
qumode-coupled evolution does realise $e^{At}$ on a subspace of the enlarged
space, so the oscillator may be read as a continuous-variable block encoding
of the propagator; the distinction we draw is that it is \emph{constructed
and priced} rather than \emph{assumed and queried}.  Two inputs remain
outside this accounting: preparation of the amplitude-encoded
$|\boldsymbol{\rho}(0)\rangle$, efficient for approximate Gaussians and
polynomial functions~\cite{XieBenAmi2025,GonzalezConde2024}, and preparation
of the qumode kernel state, which we do cost, in Fock cutoff.

\Cref{sec:related} places the algorithm in the literature.
\Cref{sec:setup,sec:fokkerplanck,sec:disc} set up the problem, the
linearisation, and the discretisation with its split $A=H_{1}-iH_{2}$, and
\cref{sec:schrod} the Schr\"{o}dingerisation and the abscissa bound.
\Cref{sec:pauli} establishes the structural results,
\cref{sec:implementation} the circuit and the algorithm,
\cref{sec:complexity} the resource analysis, and \cref{sec:numerics} the
numerical validation; \cref{sec:discussion,sec:conclusion} discuss
extensions and conclude.  Two appendices prove the abscissa bound and the
longer structural results.

\section{Related work}\label{sec:related}

Quantum algorithms for differential equations differ along two axes: how a
nonlinear system is reduced to a linear one, and how the resulting linear,
generically non-unitary, evolution is realised by unitary gates.

On the first axis, the prevailing route is \emph{Carleman linearisation},
lifting a polynomial ODE into a truncated moment hierarchy, with rigorous
bounds under a dissipativity condition~\cite{Liu2021} and in a
non-dissipative resonant regime~\cite{WuWangLi2025}; its dimension grows
combinatorially with the truncation order and its convergence is governed by
the ratio of nonlinear to linear forcing.  \emph{Koopman--von Neumann} is a
third route, used on continuous-variable hardware by Gan et
al.~\cite{GanEtAl2025}, who carry the field in multimode coherent-state
amplitudes and realise the lifted generator as a completely positive map,
under a positive semi-definite generator and post-selection at every step.
The \emph{Fokker--Planck}
route~\cite{TennieMagri2024,TennieMagri2025,Risken1989} adopted here instead
lifts the dynamics into one linear transport--diffusion PDE for the density,
exact for any smooth drift and indifferent to the balance of nonlinear and
linear forcing, at the cost of a discretised density on $\R^{d}$ and a
diffusion term whose norm grows with resolution.  It returns a distribution
rather than the trajectory alone, which costs more when only the trajectory
is wanted and supplies more when the statistics are the object of interest.

On the second axis, given $\dot{\boldsymbol{\rho}}=A\boldsymbol{\rho}$ the
quantum linear-systems and Hamiltonian-simulation
backends~\cite{Harrow2009,Childs2021,Krovi2023,Chakraborty2024} apply, and
three unitary embeddings of a non-unitary $A$ are standard:
\emph{Schr\"{o}dingerisation}~\cite{JinLiu2022,JinLiu2024,JinLiuYu2023},
which we use and which requires semi-stability, guaranteed here by the
Markov structure of the discretisation;
\emph{LCHS}~\cite{AnLiuLin2023,AnChildsLin2023}; and
\emph{time-marching}~\cite{FangLinTong2023}.  Li~\cite{Li2026Dilation} shows
these to be members of one moment-matching dilation family whose
tight-binding representative places the ancilla on a one-dimensional
lattice, a framework that classifies our choice rather than competing with
it, our qumode being the continuum limit of that lattice.  Closely related
are dilations of the Lindblad equation~\cite{DingLiLin2024}, and the
fast-forwarding limitations of~\cite{AnLiuWangZhao2025}, which govern our
post-selection cost.  Schr\"{o}dingerisation and LCHS both sweep a continuum
in one scalar parameter; our departure is to carry it on a physical mode.
Since our structural results concern only the split $A=H_{1}-iH_{2}$, any
dilation coupling an ancilla to that split inherits them.  The stronger
condition $H_{1}\preceq0$, under which recovery may use the whole half-line
$\xi>0$, fails for a non-constant drift; Jin, Liu, and
Ma~\cite{JinLiuMa2025,JinLiuMaIllPosed2025} show that a system carrying
unstable modes is nonetheless recovered from suitable data in the extended
space, and we identify that data here and bound
it uniformly in the grid.

Nearest to this work are the continuous-variable realisations.  Jin and
Liu~\cite{JinLiuAnalogPDE2024,JinLiuJC2024,JinLiuCV2025,JinLiuYu2023} map a
$D$-dimensional linear PDE onto $D{+}1$ qumodes with \emph{no} spatial
discretisation, using couplings available in cavity and circuit QED, and Das
et al.~\cite{DasEtAl2026} price hybrid oscillator--qubit kernel states for
the analogous LCHS representation, whose estimates we adopt.  We differ from
the analog programme in which degrees of freedom are carried by the
continuum: there the $D$ spatial variables and the Schr\"{o}dingerisation
variable all are, here only the latter, the density staying on an $n$-qubit
register.  Keeping that register digital is what exposes the Pauli structure
we exploit, a statement about a finite-difference stencil on
$2^{n}$ points with no counterpart when $x$ is encoded in the quadratures;
it also weakens the requirement that $\mathbf{f}$ be realisable as a continuous-variable
Hamiltonian to evaluability on the grid, at the price of discretisation
error and a per-step depth polynomial in $\log N$.

For the Fokker--Planck equation specifically, Jin, Liu, and
Yu~\cite{JinLiuYuFP2024} give a detailed Schr\"{o}dingerisation algorithm,
with circuits built by time splitting from shift operators diagonalised in
the Fourier and Bell bases, and~\cite{HuJinLiuZhang2024} gives
qubit-register realisations with explicit gate counts; Kharazi et
al.~\cite{Kharazi2026} prove speedups for reaction rates by a Gaussian-LCHS
representation on a block encoding, paired with a generalised Hadamard test.
Both target settings narrower than ours in one respect and are complementary
in another.  Their drifts are potentials ($\nabla V$) or reversible
gradient-drifts with a Boltzmann equilibrium, which makes the
generator symmetrisable, and their boundaries are periodic, whereas ours is a
general polynomial $\mathbf{f}$, which for $d\ge2$ need not be a gradient, on
a conservative reflecting wall dictated by the problem, \cref{eq:ode_intro}
being posed on $\R^{d}$ with trajectories that do not wrap.  The wall
separates the constructions twice: it makes the shifts non-circulant,
removing the fixed diagonalising basis periodicity supplies, and it puts two
non-interior rows into $H_{1}$ whose Gershgorin discs reach
$\Order(1/\Delta x)$, so a grid-independent abscissa requires
\cref{thm:abscissa}.  Where~\cite{Kharazi2026} computes a scalar rate
from a block encoding whose own cost is left outside the accounting, we
return the full density from a circuit compiled directly from the drift
coefficients.

The tool that makes that compilation possible is the tridiagonal Pauli
decomposition of Arseniev et al.~\cite{Arseniev2024}, which supplies the
bipartite structure and the commuting-family count we use; our
contribution is what a polynomial drift adds to it.  Related encodings of
finite-difference operators, under periodic, Dirichlet, Neumann, and Robin
conditions~\cite{Kharazi2025} and with physical boundary and interface
conditions within Schr\"{o}dingerisation~\cite{JinLiLiuYu2024}, are
general-purpose where ours is specific to this stencil, and our sparsity
bound also sharpens generic Pauli-partitioning
heuristics~\cite{Crawford2021,Reggio2024} for this operator class.

\section{Problem setup}\label{sec:setup}

We consider an autonomous nonlinear ODE
\begin{equation}
  \dot{\mathbf{x}}(t) \;=\; \mathbf{f}(\mathbf{x}(t)),\qquad
  \mathbf{x}(t)\in\R^{d},\quad
  \mathbf{x}(0)=\mathbf{x}_{0},
  \label{eq:ode}
\end{equation}
where the drift $\mathbf{f}\colon\R^{d}\to\R^{d}$ is smooth.  Throughout
the paper we make the following \emph{polynomial-smoothness} assumption:
each component $f_{i}$ can be uniformly approximated, on the spatial
domain of interest, by a polynomial of total degree at most~$L$
\emph{jointly} in the $d$ coordinates $x_{1},\dots,x_{d}$ (not merely degree $L$ per axis).

\section{Fokker--Planck linearisation}\label{sec:fokkerplanck}

We embed \cref{eq:ode} into a stochastic framework by appending a small
isotropic Wiener forcing,
\begin{equation}
  d\mathbf{x}
  \;=\; \mathbf{f}(\mathbf{x})\,dt
   \;+\;\sqrt{2\sigma}\,d\mathbf{W}_{t},\qquad \sigma>0,
  \label{eq:sde}
\end{equation}
with $\mathbf{W}_{t}$ a standard $d$-dimensional Wiener process.  In the
zero-noise limit $\sigma\to0$ the trajectories of \cref{eq:sde} converge
to those of \cref{eq:ode}, so the embedding is a regularisation rather
than an approximation~\cite{Risken1989}.

The PDF $\rho(\mathbf{x},t)$ of the diffusion \cref{eq:sde} satisfies the
\emph{Fokker--Planck equation},
\begin{equation}
  \partial_{t}\rho
   \;=\; -\sum_{i=1}^{d}\partial_{x_{i}}\!\bigl[f_{i}(\mathbf{x})\,\rho\bigr]
   \;+\;\sigma\,\Delta_{\mathbf{x}}\rho
   \;\equiv\; A_{\mathrm{FP}}\,\rho,
  \label{eq:FP}
\end{equation}
which follows from It\^{o}'s lemma and the Kramers--Moyal expansion (the
higher Kramers--Moyal coefficients vanish by Pawula's
theorem~\cite{Pawula1967,Risken1989}).  Although the drift
$\mathbf{f}$ is nonlinear in $\mathbf{x}$, equation \cref{eq:FP} is
\emph{linear} in the unknown $\rho$.  The operator $A_{\mathrm{FP}}$ is
non-Hermitian: the diffusion is self-adjoint and negative semi-definite,
while the drift term $-\nabla\!\cdot\!(\mathbf{f}\,\cdot)$ decomposes
into a skew-adjoint advection and the symmetric multiplication operator
$-\tfrac12(\nabla\!\cdot\!\mathbf{f})$. The linearity of $A_{\mathrm{FP}}$ with respect to $\rho$ makes it a
suitable target for a linear-systems quantum algorithm.

\begin{remark}\label{rem:recovery}
  The deterministic trajectory of \cref{eq:ode} is recovered
  as the peak(s) of $\rho$ in the limit $\sigma\to0$;
  physical observables of interest are typically linear or low-order polynomial
  functionals of $\rho$.
\end{remark}

\section{Spatial discretisation and Hermitian decomposition}\label{sec:disc}\label{sec:cd}\label{sec:multidim}\label{sec:herm}

We present $d=1$; the extension to $d$ dimensions is by tensor product
below.  On a uniform grid $x_{j}=x_{\min}+j\,\Delta x$, $j=0,\dots,N-1$ with
$\Delta x=(x_{\max}-x_{\min})/(N-1)$, writing
$\rho_{j}(t)\approx\rho(x_{j},t)$, $f_{j}=f(x_{j})$ and
$\boldsymbol{\rho}=(\rho_{0},\dots,\rho_{N-1})^{T}$, second-order central
differences turn \cref{eq:FP} into
$\dot{\boldsymbol{\rho}}=A\boldsymbol{\rho}$ with $A\in\R^{N\times N}$ real,
generically non-symmetric, and tridiagonal:
\begin{align}
  A_{j,j-1} &= \frac{f_{j-1}}{2\Delta x}+\frac{\sigma}{\Delta x^{2}}, &
  A_{j,j}   &= -\frac{2\sigma}{\Delta x^{2}}, &
  A_{j,j+1} &= -\frac{f_{j+1}}{2\Delta x}+\frac{\sigma}{\Delta x^{2}}.
  \label{eq:A_entries}
\end{align}
The CFL-type condition $\Delta x\le\min_{j}2\sigma/|f_{j}|$ makes the
off-diagonals non-negative and $e^{At}$ a stochastic
semigroup~\cite{Holubec2019}.  Reflecting boundaries zero the entries
reaching outside the domain, which is not exactly conservative, the two
boundary columns then failing to sum to zero; we restore
$\mathbf{1}^{T}A=0$ by the \emph{master-equation form}, setting each boundary
diagonal to minus the sum of its retained outgoing rates.  This leaves the
interior stencil unchanged and is a rank-two correction,
\begin{equation}
  A \;=\; A_{\mathrm{bulk}} \;+\; \Delta A,\qquad
  \Delta A \;=\; \delta_{0}\,|0\rangle\langle 0|
                 \;+\;\delta_{N-1}\,|N\!-\!1\rangle\langle N\!-\!1|,
  \label{eq:A_split}
\end{equation}
with $\delta_{0}=\sigma/\Delta x^{2}-f_{0}/(2\Delta x)$ and
$\delta_{N-1}=\sigma/\Delta x^{2}+f_{N-1}/(2\Delta x)$.  The structural analysis below is
carried out for $A_{\mathrm{bulk}}$; $\Delta A$ is treated in
\cref{rem:boundary}.

\begin{remark}[Periodic boundaries]\label{rem:periodic}
  Nothing requires the reflecting wall, and periodicity is the simpler case.
  The wrap-around entries are completed from \cref{eq:A_entries} with
  $f_{N}\equiv f_{0}$; conservation is then automatic, every cyclic column
  already summing to zero, so $\Delta A$ vanishes.  The wrap bond joins
  $|0\cdots0\rangle$ to $|1\cdots1\rangle$, which differ in every bit, so it
  enters the existing carry level $m=n$ and populates no new family, at one
  additional controlled displacement per Trotter step.  We therefore carry
  the reflecting case throughout.
\end{remark}

\paragraph{Multi-dimensional extension}

For $d$ dimensions with $N_{k}=2^{n_{k}}$ points along axis $k$, so
$N_{\mathrm{tot}}=\prod_{k}N_{k}=2^{n_{\mathrm{tot}}}$ on
$n_{\mathrm{tot}}=\sum_{k}n_{k}$ qubits, the discretised operator splits into
one term per axis, each a \emph{transverse-projector-weighted} sum of
one-dimensional tridiagonal blocks, one per grid line parallel to that axis
(\cite{TennieMagri2025}, eq.~(3.7)):
\begin{equation}
  A \;=\; \sum_{k=1}^{d} A^{(k)},\qquad
  A^{(k)} \;=\; \sum_{\mathbf{p}_{\perp}}
    E^{(N_{1})}_{p_{1}p_{1}} \otimes\cdots\otimes
    \underbrace{A_{k}(\mathbf{p}_{\perp})}_{\text{slot }k}
    \otimes\cdots\otimes E^{(N_{d})}_{p_{d}p_{d}},
  \label{eq:A_decomp}
\end{equation}
where $\mathbf{p}_{\perp}$ runs over the transverse multi-index,
$E^{(N_{i})}_{pp}$ projects onto index $p$, and $A_{k}(\mathbf{p}_{\perp})$
is the stencil \cref{eq:A_entries} for axis $k$ assembled from $f_{k}$ on the
grid line whose transverse coordinates are frozen at $\mathbf{p}_{\perp}$.
Since those coefficients may depend on all $d$ coordinates the decomposition
covers fully coupled drifts, collapsing to a bare tensor product only in the
separable case $f_{k}(\mathbf{x})=f_{k}(x_{k})$.  The transverse factors are
computational-basis projectors, hence diagonal, and this is what preserves
the bipartite Pauli structure in $d$ dimensions (\cref{rem:multidim_pauli}).

\paragraph{Hermitian decomposition}
Any real $A$ splits uniquely as
\begin{equation}
  A \;=\; H_{1} - i H_{2},\qquad
  H_{1}=\tfrac{A+A^{T}}{2},\qquad
  H_{2}=\tfrac{A^{T}-A}{2i},
  \label{eq:A_decomp_herm}
\end{equation}
with $H_{1}$ symmetric, carrying the diffusion and a symmetrised drift
correction, and $H_{2}$ skew, carrying the advection.  From
\cref{eq:A_entries}, $H_{1}$ is real symmetric tridiagonal with constant
diagonal $-2\sigma/\Delta x^{2}$ and off-diagonals
\begin{equation}
  \alpha_{j}\;:=\;\frac{\sigma}{\Delta x^{2}}+\frac{f_{j}-f_{j+1}}{4\Delta x}
  \;\in\;\R,
  \label{eq:alpha_def}
\end{equation}
while $H_{2}$ is Hermitian tridiagonal with vanishing diagonal and
off-diagonals $(H_{2})_{j,j+1}=-i\beta_{j}=\overline{(H_{2})_{j+1,j}}$,
\begin{equation}
  \beta_{j}\;:=\;\frac{f_{j}+f_{j+1}}{4\Delta x}\;\in\;\R,
  \label{eq:beta_def}
\end{equation}
the advective rate on the bond, of either sign.  Introducing for each bond
$j\to j+1$
\begin{equation}
  S_{j}^{+} := |j\rangle\langle j+1| + |j+1\rangle\langle j|,\qquad
  S_{j}^{-} := i\bigl(|j\rangle\langle j+1|-|j+1\rangle\langle j|\bigr),
  \label{eq:Spm}
\end{equation}
we may write compactly
\begin{equation}
  H_{1} \;=\; -\frac{2\sigma}{\Delta x^{2}}\,\Id
              \;+\;\sum_{j=0}^{N-2}\alpha_{j}\,S_{j}^{+},\qquad
  H_{2} \;=\; -\sum_{j=0}^{N-2}\beta_{j}\,S_{j}^{-}.
  \label{eq:H1H2_bond}
\end{equation}
The bond operators \cref{eq:Spm} are the building blocks of the Pauli
analysis that follows.

\section{Schr\"{o}dingerisation}\label{sec:schrod}\label{sec:recovery_domain}

Schr\"{o}dingerisation~\cite{JinLiu2022,JinLiu2024,JinLiuYu2023} embeds
$e^{At}$ into a unitary flow on one added dimension.  With $\xi\in\R$ and the
warped-phase variable $w(t,\xi):=e^{-\xi}\boldsymbol{\rho}(t)$ for $\xi>0$,
extended by $w(0,\xi)=e^{-|\xi|}\boldsymbol{\rho}(0)$, the identity
$\partial_{\xi}w=-w$ on $\xi>0$ turns $\partial_{t}w=Aw$ into
\begin{equation}
  \partial_{t}w \;=\; -H_{1}\,\partial_{\xi}w \;-\; i\,H_{2}\,w ,
  \label{eq:warped_pde}
\end{equation}
and the Fourier transform
$\tilde{w}(t,\eta)=\int_{\R}e^{-2\pi i\eta\xi}w\,d\xi$, sending
$-\partial_{\xi}\mapsto-2\pi i\eta$, gives a parametrised family of
Schr\"{o}dinger equations
\begin{equation}
  i\,\partial_{t}\tilde{w}
  \;=\;\mathcal{H}(\eta)\,\tilde{w},\qquad
  \mathcal{H}(\eta) \;:=\; 2\pi\eta\,H_{1}+H_{2},\qquad
  \eta\in\R,
  \label{eq:schrodinger_family}
\end{equation}
each generator Hermitian, so
$\tilde{w}(t,\eta)=e^{-i\mathcal{H}(\eta)t}\tilde{w}(0,\eta)$.  The Fourier
transform of $e^{-|\xi|}$ is the Lorentzian
\begin{equation}
  \widehat{e^{-|\cdot|}}(\eta)
  \;=\;\frac{2}{1+4\pi^{2}\eta^{2}},
  \label{eq:fourier_initial}
\end{equation}
so $\tilde{w}(0,\eta)=2(1+4\pi^{2}\eta^{2})^{-1}\boldsymbol{\rho}(0)$; this
is the qumode wavefunction used below.

\paragraph{Semi-stability and the recovery domain}
Two spectral quantities must be distinguished: the \emph{spectral abscissa}
$\alpha(A)=\max_{i}\operatorname{Re}\lambda_{i}(A)$, governing
semi-stability, and the \emph{numerical abscissa}
$\mu(A)=\lambda_{\max}(H_{1})$, the logarithmic norm, governing
$\|e^{At}\|_{2}\le e^{\mu(A)t}$; they coincide only for normal $A$, and this
generator is non-normal~\cite{TennieMagri2025}.  The positivity condition
and \cref{eq:A_split} make $A$ a Markov generator with $\mathbf{1}^{T}A=0$,
so $\alpha(A)\le0$ and Schr\"{o}dingerisation
applies~\cite{TennieMagri2025,JinLiuYuFP2024}.  The numerical abscissa,
however, is generically positive: in the continuum
\begin{equation}
  H_{1}
  \;=\;\tfrac12\bigl(A_{\mathrm{FP}}+A_{\mathrm{FP}}^{\dagger}\bigr)
  \;=\;\sigma\Delta \;-\;\tfrac12(\nabla\!\cdot\!\mathbf{f}),
  \label{eq:H1_continuum}
\end{equation}
a Schr\"{o}dinger-type operator whose potential is positive wherever the
flow compresses.  This is handled by shifting the region from which the
original variable is recovered, as Jin, Liu, and Ma~\cite{JinLiuMa2025}
establish for Schr\"{o}dingerised systems carrying unstable modes.  For the
present generator the shift is set by $\lambda_{\max}(H_{1})$, which the
following bounds uniformly in the grid; Gershgorin does not suffice, since
the two wall rows of the conservative discretisation have discs reaching
$\Order(1/\Delta x)$ under an outward drift, while the boundary layer that
produces the growth has width set by $\sigma$ and the drift, not by
$\Delta x$.

\begin{theorem}[Grid-independent numerical abscissa]\label{thm:abscissa}
  Let $f\in C^{1}([x_{\min},x_{\max}])$, and set
  $\Lambda:=\max_{x}|f'(x)|$, $F:=\max_{x}|f(x)|$,
  $M:=F_{\partial}/(2\Delta x)+\Lambda/4$ and
  \begin{equation}
    F_{\partial}\;:=\;
      \max\bigl(\,[-f(x_{\min})]_{+},\;[\,f(x_{\max})]_{+}\bigr)
    \label{eq:Fbdry}
  \end{equation}
  for the outward drift speed at the walls.  If
  \begin{equation}
    \Delta x\;\le\;\frac{2\sigma}{F},\qquad
    \Delta x^{2}\;\le\;\frac{2\sigma}{\Lambda},\qquad
    4M\,\Delta x^{2}\;\le\;\sigma,
    \label{eq:abscissa_hyp}
  \end{equation}
  then the symmetric part of the master-equation generator satisfies
  \begin{equation}
    \lambda_{\max}(H_{1})
    \;\le\;\tfrac12\max_{x}\bigl(-f'(x)\bigr)
      \;+\;\frac{32\,M^{2}\Delta x^{2}}{\sigma}
      \;+\;\Order(\Delta x^{2}).
    \label{eq:abscissa}
  \end{equation}
  Since $M\Delta x=F_{\partial}/2+\Lambda\Delta x/4$, the second term tends
  to $8F_{\partial}^{2}/\sigma$ as $\Delta x\to0$, so $\lambda_{\max}(H_{1})$
  is bounded uniformly in $\Delta x$ by a constant fixed by the drift and the
  diffusion.
\end{theorem}

\begin{proof}[Proof sketch]
  The quadratic form collapses onto the bonds,
  \begin{equation}
    \boldsymbol{\rho}^{\!\top}\!H_{1}\boldsymbol{\rho}
    \;=\;\sum_{j=0}^{N-2}\Bigl[
        -\alpha_{j}(\rho_{j+1}-\rho_{j})^{2}
        +\beta_{j}\bigl(\rho_{j+1}^{2}-\rho_{j}^{2}\bigr)\Bigr],
    \label{eq:quadform}
  \end{equation}
  and summation by parts sends the second sum to
  $\tfrac12\max_{x}(-f'(x))\|\boldsymbol{\rho}\|^{2}$ plus two wall terms.  A
  discrete trace inequality bounds each wall value by a multiple of the
  Dirichlet sum plus a multiple of $\|\boldsymbol{\rho}\|^{2}$, and
  \cref{eq:abscissa_hyp} gives $\alpha_{j}\ge\sigma/(2\Delta x^{2})$, enough
  to absorb the Dirichlet part.  \Cref{app:abscissa} has the details.
\end{proof}

Since \cref{eq:warped_pde} is symmetric hyperbolic with propagation speeds
the eigenvalues of $H_{1}$, the solution at $(t,\xi)$ depends on the data
only on $[\xi-\lambda_{\max}t,\infty)$, with
$\lambda_{\max}:=\lambda_{\max}(H_{1})$.  For $\xi>\max(\lambda_{\max},0)t$
that interval lies in $\xi'>0$, where
$w(0,\xi')=e^{-\xi'}\boldsymbol{\rho}(0)$ exactly and the even extension
never reaches, so
\begin{equation}
  \boldsymbol{\rho}(t)\;=\;e^{\xi^{*}}\,w(t,\xi^{*})
  \qquad\text{for any }\xi^{*}>\lambda_{\max}t.
  \label{eq:shifted_recovery}
\end{equation}
When $\lambda_{\max}>0$ the usable amplitude is reduced by
$e^{-\lambda_{\max}t}$, a factor carried through the post-selection cost and
consistent with the lower bounds of~\cite{AnLiuWangZhao2025}.  The unshifted
projection onto $\xi>0$ used in~\cite{JinLiu2022} is correct under the
hypothesis stated there, which in our convention is $H_{1}\preceq0$; that
fails here, $H_{1}$ being indefinite.

\section{Pauli structure of \texorpdfstring{$H_{1}$ and $H_{2}$}{H1 and H2}}\label{sec:pauli}\label{sec:pauli_conv}\label{sec:bipartite}\label{sec:weight_conc}

Take $N=2^{n}$ and store the probability vector on $n$ qubits as
\begin{equation}
  |\boldsymbol{\rho}(t)\rangle
  \;=\;\frac{1}{\|\boldsymbol{\rho}(t)\|}\sum_{j=0}^{N-1}\rho_{j}(t)\,|j\rangle,
  \label{eq:quantum_state}
\end{equation}
the normalisation being recovered at readout from the post-selection
statistics.  Any Hermitian $N\times N$ matrix $M$
expands as
\begin{equation}
  M \;=\;\sum_{\mathbf{k}\in\{0,1,2,3\}^{n}}
          c_{\mathbf{k}}\,P_{\mathbf{k}},\qquad
  c_{\mathbf{k}}=\tfrac{1}{N}\,\tr(M P_{\mathbf{k}})\in\R,
  \label{eq:pauli_decomp}
\end{equation}
with $P_{\mathbf{k}}=\sigma_{k_{1}}\otimes\cdots\otimes\sigma_{k_{n}}$ and
$(\sigma_{0},\sigma_{1},\sigma_{2},\sigma_{3})=(I,X,Y,Z)$.  Write
\begin{equation}
  H_{1}=\sum_{k}c^{(1)}_{k}\,P_{k},
  \qquad
  H_{2}=\sum_{k}c^{(2)}_{k}\,Q_{k},
  \label{eq:H12_pauli}
\end{equation}
where $\{P_{k}\}$ and $\{Q_{k}\}$ enumerate the strings with non-zero
coefficient in $H_{1}$ and $H_{2}$; the two sets are disjoint by \cref{thm:chromatic}.  Then \cref{eq:schrodinger_family} reads
\begin{equation}
  \mathcal{H}(\eta)
  \;=\;\sum_{k}\bigl(2\pi\eta\,c^{(1)}_{k}\bigr)\,P_{k}
   \;+\;\sum_{k}c^{(2)}_{k}\,Q_{k}.
  \label{eq:H_eta_pauli}
\end{equation}

\paragraph{Pauli structure of tridiagonal generators}

Adding $1$ to a bond index $j$ flips a run of low-order bits.  If the
least-significant $0$-bit of $j$ lies at position $k$, call $m:=k+1$ the
\emph{carry length}, so that
\[
  j \;=\; b_{n-1}\cdots b_{m}\,\underbrace{0\,1\cdots 1}_{m},
  \qquad
  j+1 \;=\; b_{n-1}\cdots b_{m}\,\underbrace{1\,0\cdots 0}_{m}:
\]
the upper $n-m$ bits (the \emph{prefix}) agree and the lower $m$ (the
\emph{suffix}) are all flipped.  Hence \cref{eq:Spm} factorises as
\begin{equation}
  S_{j}^{\pm}
  \;=\;|p\rangle\langle p|_{\mathrm{pre}}\otimes\Sigma_{m}^{\pm},
  \label{eq:bond_factor}
\end{equation}
with $p$ the common prefix and
\begin{align}
  \Sigma_{m}^{+} &:= |s_{m}\rangle\langle s_{m}'|+|s_{m}'\rangle\langle s_{m}|,
    \label{eq:Sigmaplus}\\
  \Sigma_{m}^{-} &:= i\bigl(|s_{m}\rangle\langle s_{m}'|-|s_{m}'\rangle\langle s_{m}|\bigr),
    \label{eq:Sigmaminus}
\end{align}
where $|s_{m}\rangle=|0\,1\cdots1\rangle$ and $|s_{m}'\rangle=|1\,0\cdots0\rangle$.

\begin{proposition}[Bipartite structure and commuting families;
  Arseniev et al.~\cite{Arseniev2024}]
  \label{thm:pauli_structure}
  Let $B\in\C^{N\times N}$, $N=2^{n}$, be tridiagonal.  Every non-zero Pauli
  string in the decomposition \cref{eq:pauli_decomp} of $B$ lies in one of
  the $n+1$ disjoint sets
  \begin{equation}
    S^{n}_{m}\;:=\;\{I,Z\}^{\otimes(n-m)}\otimes\{X,Y\}^{\otimes m},
    \qquad m=0,\dots,n,
    \label{eq:bipartite}
  \end{equation}
  indexed by the carry length $m$: the prefix qubits carry only $\{I,Z\}$
  and the suffix qubits only $\{X,Y\}$.  Encoding a string by
  $(x,z)\in\{0,1\}^{n}\times\{0,1\}^{n}$ through the Walsh function
  \begin{equation}
    \hat{W}(x,z)
    \;:=\; i^{\,x\cdot z}\,\bigotimes_{i=1}^{n} X_{i}^{x_{i}}Z_{i}^{z_{i}},
    \label{eq:walsh}
  \end{equation}
  two strings $\hat{W}(x,z)$, $\hat{W}(a,b)$ commute iff
  \begin{equation}
    x\cdot b \;\equiv\; a\cdot z \pmod{2};
    \label{eq:commute_criterion}
  \end{equation}
  in particular strings sharing an $x$-vector and a $Y$-parity commute.
  Every string of $S^{n}_{m}$ has $x=V_{m}:=(0,\dots,0,1,\dots,1)$ with $m$
  trailing ones, so splitting by $Y$-parity,
  \begin{align}
    S_{m,+} &:= \bigl\{\hat{W}(V_{m},z)\;:\; V_{m}\cdot z\equiv 0\bigr\},
      \label{eq:Smp}\\
    S_{m,-} &:= \bigl\{\hat{W}(V_{m},z)\;:\; V_{m}\cdot z\equiv 1\bigr\},
      \label{eq:Smm}
  \end{align}
  each subset is internally commuting, and a real tridiagonal $B$ populates
  at most $2n+1$ of them.
\end{proposition}

\noindent These are Propositions~1--3 and Corollary~1 of~\cite{Arseniev2024}, restated
in the notation above; we claim no novelty in them.  The mechanism is
\cref{eq:bond_factor}: the prefix projector expands in $\{I,Z\}$ alone,
\begin{equation}
  |p\rangle\langle p|
  \;=\;\frac{1}{2^{n-m}}\bigotimes_{i=m}^{n-1}
        \bigl(\Id+(-1)^{p_{i}}\,Z_{i}\bigr),
  \label{eq:prefix_proj}
\end{equation}
and is the identity on the suffix, while $\Sigma_{m}^{\pm}$ connects two
states differing in every suffix slot and so expands on
$\{X,Y\}^{\otimes m}$, since $\langle s|P|s'\rangle$ vanishes unless
$P_{i}\in\{I,Z\}$ wherever $s_{i}=s_{i}'$ and $P_{i}\in\{X,Y\}$ wherever
$s_{i}\neq s_{i}'$.  What the Fokker--Planck setting adds is the assignment
of the two parity classes to the two Hermitian parts, and the extension to
the projector-weighted decomposition \cref{eq:A_decomp} in $d$ dimensions.

\begin{corollary}[Commuting families of the Hermitian split]
  \label{thm:chromatic}
  For the generator \cref{eq:A_entries} the non-zero strings satisfy
  \begin{equation}
    H_{1}:\;\{S_{0,+},S_{1,+},\dots,S_{n,+}\},
    \qquad
    H_{2}:\;\{S_{1,-},S_{2,-},\dots,S_{n,-}\},
    \label{eq:H1_families}
  \end{equation}
  so $G_{\mathrm{total}}\le2n+1=\Order(\log N)$.  Under \cref{eq:A_decomp}
  each axis contributes its own families and
  $G_{\mathrm{total}}\le\sum_{k}(2n_{k}+1)=2n_{\mathrm{tot}}+d$, linear in
  the total qubit count rather than $d$ times it.
\end{corollary}

\begin{proof}
  $A$ is real tridiagonal, so \cref{thm:pauli_structure} applies and at most
  $2n+1$ subsets are populated; which ones follows from transposition rather
  than Hermiticity, the coefficients \cref{eq:pauli_decomp} being real for
  either part.  Since $Y^{T}=-Y$ while $I,X,Z$ are symmetric,
  \begin{equation}
    P^{T}=(-1)^{w(P)}\,P,\qquad w(P):=x\cdot z\bmod 2,
    \label{eq:pauli_transpose}
  \end{equation}
  so even-$Y$ strings are symmetric matrices and odd-$Y$ strings
  antisymmetric.  Transposing \cref{eq:H12_pauli} and invoking linear
  independence of the Pauli basis, the symmetry of $H_{1}=(A+A^{T})/2$ kills
  every odd-$Y$ coefficient and the antisymmetry of $H_{2}=i(A^{T}-A)/2$
  every even-$Y$ one, the vanishing diagonal of $H_{2}$ excluding $m=0$,
  which is \cref{eq:H1_families}; real coefficients and an imaginary $H_{2}$
  are consistent because the odd-$Y$ strings are themselves imaginary
  antisymmetric matrices.  In $d$ dimensions the transverse factors of
  \cref{eq:A_decomp} are diagonal projectors, so an axis-$k$ string carries
  $\{I,Z\}$ on every transverse register and its $x$-vector is $V_{m}$
  supported on register $k$; strings sharing an axis and carry level
  therefore commute by \cref{eq:commute_criterion}, and distinct axes give
  disjoint families (\cref{rem:multidim_pauli}).
\end{proof}

\paragraph{Sparsity per family}

The family count is independent of the drift degree $L$.  What $L$ controls,
and what is specific to a polynomial drift rather than to a general
tridiagonal matrix, is how many strings within a family are populated: it
bounds the number of non-zero Pauli terms, though not the rank of the
subgroup they generate.  This subsection and the next are our contribution.
For fixed $m$, a string $\hat{W}(V_{m},z)\in S_{m,+}$ is populated iff the
corresponding prefix coefficient of
$\sum_{p\in\mathcal{P}_{m}}\alpha_{j(p,m)}|p\rangle\langle p|$ is non-zero,
the suffix contributing all $2^{m-1}$ strings of $\{X,Y\}^{\otimes m}$ of one
$Y$-parity.  Sparsity is therefore controlled by the prefix, through the
following property of a uniform grid.
\begin{lemma}[Affine grid coordinate]\label{lem:affine}
  On $x_{j}=x_{\min}+j\,\Delta x$ the coordinate is affine in the bits of
  $j=\sum_{i}b_{i}2^{i}$,
  \begin{equation}
    x_{j}\;=\;x_{\min}+\Delta x\sum_{i=0}^{n-1}2^{i}\,b_{i}.
    \label{eq:affine_grid}
  \end{equation}
  Hence for any polynomial $P$ of degree $\le L$, $\operatorname{diag}(P(x_{j}))$
  is multilinear of degree $\le L$ in $\hat n_{i}=(\Id-Z_{i})/2$, and its
  Pauli expansion is supported on $Z$-strings of Hamming weight $\le L$.
\end{lemma}

\begin{proof}
  Raising \cref{eq:affine_grid} to a power $\ell\le L$ and reducing
  $b_{i}^{2}=b_{i}$ gives multilinear degree $\le\ell$; substituting
  $b_{i}=\hat n_{i}$, each monomial $\prod_{i\in S}b_{i}$ with $|S|\le L$
  contributes only $Z$-strings supported on subsets of
  $S$~\cite{Odonnell2014}.
\end{proof}

The map $p\mapsto j(p,m)=p\,2^{m}+(2^{m-1}-1)$ is affine, so
\cref{lem:affine} applies on the $(n-m)$-qubit prefix and the Walsh support
of $\alpha_{j(p,m)}$, and likewise of $\beta_{j(p,m)}$, has Hamming weight
$\le L$, of which there are at most
\begin{equation}
  \sum_{\ell=0}^{L}\binom{n-m}{\ell}
  \;\le\;\sum_{\ell=0}^{L}\binom{n}{\ell}
  \;=\;\Order(n^{L}).
  \label{eq:prefix_count}
\end{equation}

\begin{theorem}[Pauli sparsity per family]
  \label{thm:pauli_count}
  Let $f$ be polynomial of degree $\le L$ on the grid of \cref{lem:affine}.
  For each carry level $m$, the number of non-zero Pauli terms of $H_{1}$ or
  $H_{2}$ in $S_{m,\pm}$ is at most
  \begin{equation}
    2^{\max(m-1,0)}\sum_{\ell=0}^{L}\binom{n-m}{\ell}\;=\;\Order(2^{m}n^{L}),
    \label{eq:term_count}
  \end{equation}
  namely \cref{eq:prefix_count} times the $2^{m-1}$ suffix strings of one
  $Y$-parity ($1$ for $m=0$).  If $f$ is instead a degree-$L$ approximant of
  a smooth drift with sup-norm error $\delta$, the residual contributes Walsh
  coefficients of magnitude $\Order(\delta/\Delta x)$.
\end{theorem}

\Cref{eq:term_count} is saturated for a generic degree-$L$ drift; symmetries
of the drift or of the grid, for instance a grid centred on an axis of
symmetry of $f$, annihilate further Walsh coefficients and make the affected
families strictly sparser.

\begin{proposition}[Structural results in $d$ dimensions]\label{rem:multidim_pauli}
  Under \cref{eq:A_decomp} on $n_{\mathrm{tot}}=\sum_{k}n_{k}$ qubits, with a
  uniform grid along every axis and the joint smoothness hypothesis of
  \cref{sec:setup},
  \cref{thm:pauli_structure,thm:chromatic,thm:pauli_count} hold with
  $n\to n_{\mathrm{tot}}$: every string is $\{I,Z\}$ on the transverse
  registers and the axis-$k$ prefix and $\{X,Y\}$ on the $m$-qubit axis-$k$
  suffix, the families number $2n_{\mathrm{tot}}+d$, and each carries at most
  $2^{\max(m-1,0)}\sum_{\ell\le L}\binom{n_{\mathrm{tot}}-m}{\ell}
  =\Order(2^{m}n_{\mathrm{tot}}^{L})$ terms.
\end{proposition}

The proof is in \cref{app:pauli_proofs}.  The dimension enters only
through the term count, a coupled drift spreading its Walsh support over the
transverse registers as well as the axis prefix; the family count and suffix
structure are unaffected, and that asymmetry is what sets the per-step gate count.

\begin{remark}[Boundary correction, and its cost]\label{rem:boundary}
  \Cref{thm:pauli_structure,thm:chromatic,thm:pauli_count} are stated for
  $A_{\mathrm{bulk}}$.  The correction $\Delta A$ of \cref{eq:A_split} is
  diagonal, hence lies in $H_{1}$, but its projectors expand with full
  $\{I,Z\}$ support and would populate all $2^{n}$ strings of $S_{0,+}$, so
  we keep it separate.  In one dimension \cref{eq:Htotal} then carries
  $\Delta A\otimes\qhat$, generating two momentum displacements conditioned
  through $n$-controlled gates on the all-zeros and all-ones states,
  $\Order(n)$ gates and one ancilla each and no Clifford, the projectors
  being diagonal.  In $d>1$ the wall along axis $k$ is a pair of
  $(d{-}1)$-faces and
  \begin{align}
    \Delta A
    &=\sum_{k=1}^{d}\Bigl[
        \operatorname{diag}\bigl(\delta^{(k)}_{-}\bigr)\otimes
          |0\rangle\langle0|_{k}
       +\operatorname{diag}\bigl(\delta^{(k)}_{+}\bigr)\otimes
          |N_{k}\!-\!1\rangle\langle N_{k}\!-\!1|_{k}\Bigr], \label{eq:dA_multidim}
    \\
    \delta^{(k)}_{\mp}(\mathbf{p}_{\perp})
    &=\frac{\sigma}{\Delta x_{k}^{2}}
       \mp\frac{f_{k}(x_{k,\mathrm{face}},\mathbf{p}_{\perp})}
               {2\Delta x_{k}},\nonumber
  \end{align}
  the sum over $k$ giving a site on several walls one term per wall, so edges
  and corners need no separate treatment.  Each $\delta^{(k)}_{\mp}$ has
  degree $\le L$ in the transverse multi-index, so by \cref{lem:affine} its
  Walsh support has weight $\le L$; all terms are diagonal, so
  $e^{-i\Delta A\otimes\qhat\Delta t}$ factorises exactly, at $2d$ faces carrying
  $\Order(n_{\mathrm{tot}}^{L})$ displacements of $\Order(L+n_{k})$ gates
  apiece.  The correction costs $\Order(d^{L+1}n^{L+1})$ per Trotter step, a
  factor $n$ below the bulk count, leaving the complexity analysis
  unchanged.
\end{remark}

\section{Quantum implementation}\label{sec:implementation}\label{sec:cv_lcu}\label{sec:factorisation}\label{sec:synthesis}\label{sec:qumode_prep}\label{sec:full_algorithm}

A \emph{qumode} is a quantum harmonic oscillator with $[\qhat,\phat]=i$ on
$L^{2}(\R)$, with improper position eigenstates
$\qhat|\eta\rangle_{\qhat}=\eta|\eta\rangle_{\qhat}$.  Couple the $n$-qubit
data register to one qumode through
\begin{equation}
  \Htot
  \;=\; H_{1}\otimes\qhat \;+\; H_{2}\otimes\Id_{\mathrm{osc}},
  \label{eq:Htotal}
\end{equation}
so that ${}_{\qhat}\langle\eta|\Htot|\eta\rangle_{\qhat}
=\eta H_{1}+H_{2}=\mathcal{H}(\eta)$ of
\cref{eq:schrodinger_family}, up to the $2\pi$ absorbed into the Fourier
convention.  Initialising the qumode in the Lorentzian state
\begin{equation}
  \psi(\eta) \;\propto\; \frac{2}{1+4\pi^{2}\eta^{2}},
  \label{eq:qumode_initial}
\end{equation}
matching \cref{eq:fourier_initial}, evolution under $\Htot$ gives
\begin{equation}
  |\Psi(t)\rangle
  \;=\;\int_{\R}\psi(\eta)\,e^{-i\mathcal{H}(\eta)t}\,
        |\boldsymbol{\rho}(0)\rangle\,|\eta\rangle\,d\eta,
  \label{eq:full_evolution}
\end{equation}
realising the entire $\eta$-continuum in one coherent state.

The two summands of \cref{eq:Htotal} do not commute,
$[H_{1}\otimes\qhat,H_{2}\otimes\Id]=[H_{1},H_{2}]\otimes\qhat$, so we
Trotterise,
\begin{equation}
  e^{-i\Htot\Delta t}
  \;=\; e^{-iH_{2}\otimes\Id\,\Delta t}\,
        e^{-iH_{1}\otimes\qhat\,\Delta t}
        \;+\;\Order(\Delta t^{2}),
  \label{eq:trotter1}
\end{equation}
and treat each block separately.  Block~1 is Hamiltonian simulation of
$H_{2}$ on the qubit register alone; partitioning into the families
$\{S_{m,-}\}$ of \cref{thm:chromatic} and Trotterising again,
\begin{equation}
  e^{-iH_{2}\Delta t}
  \;=\;\prod_{m=1}^{n} e^{-iH_{S_{m,-}}\Delta t}
   \;+\;\Order\!\left(\Delta t^{2}
        \sum_{m<m'}\bigl\|[H_{S_{m,-}},H_{S_{m',-}}]\bigr\|\right).
  \label{eq:trotter_H2}
\end{equation}
Block~2 acts in the qumode position basis as
$|\eta\rangle\mapsto e^{-i\eta H_{1}\Delta t}|\eta\rangle$ and decomposes
over $\{S_{m,+}\}$ identically, each factor carrying the extra tensor factor
$\qhat$.  Both blocks thus reduce to family exponentials, synthesised
exactly below at a cost set by the Pauli sparsity \cref{eq:term_count}.

\paragraph{Diagonal-block factorisation}

Fix $m$ and the $H_{1}$ family $S_{m,+}$; the $H_{2}$ case is identical with
$\beta_{j},\Sigma_{m}^{-}$ for $\alpha_{j},\Sigma_{m}^{+}$.  By
\cref{eq:bond_factor,eq:prefix_proj} the family operator is an exact tensor
product,
\begin{equation}
  H_{S_{m,+}}
  \;=\; D_{m}\otimes\Sigma_{m}^{+},
  \qquad
  D_{m} \;=\;\sum_{p\in\mathcal{P}_{m}}\alpha_{j(p,m)}\,
              |p\rangle\langle p|_{\mathrm{pre}} ,
  \label{eq:Dm_def}
\end{equation}
of a diagonal prefix operator on $n-m$ qubits and the fixed suffix operator
on $m$.  Two properties matter.  First, $\Sigma_{m}^{+}$ has rank two: it
connects the single pair $|s_{m}\rangle,|s_{m}'\rangle$, with spectrum
$\{+1,-1\}$ there and $0$ on the remaining $2^{m}-2$ states, for every $m$.
Second, $D_{m}$ is sparse in the monomial
basis: by the Walsh--Hadamard transform,
\begin{equation}
  D_{m}
  \;=\;\sum_{\substack{S\subseteq\{0,\dots,n-m-1\}\\ |S|\le L}}
        \hat{f}_{m}(S)\,\prod_{i\in S}\hat{n}_{i},
  \qquad \hat{n}_{i}:=\tfrac{\Id-Z_{i}}{2},
  \label{eq:Dm_monomials}
\end{equation}
with at most $\Order(n^{L})$ non-zero coefficients by
\cref{eq:prefix_count}.  The
$\hat{f}_{m}(S)$ are finite differences of $\alpha_{j(p,m)}$ in the prefix
index, hence computed in closed form from the drift coefficients without
assembling $A$; this is the classical preprocessing step.  All
$\hat{n}_{i}$ are diagonal and mutually commuting.

\paragraph{Exact synthesis by monomial-controlled displacements}

Take Block~2, $e^{-iH_{S_{m,+}}\otimes\qhat\,\Delta t}$; Block~1 is the same
with $\qhat\to\Id$ and displacements becoming phase rotations.  Every
summand of $H_{S_{m,+}}\otimes\qhat$ shares the factor $\Sigma_{m}^{+}$ and
differs only in commuting diagonal prefix factors, so all summands commute
and the exponential factorises exactly, with no intra-family Trotter error:
\begin{equation}
  e^{-iH_{S_{m,+}}\otimes\qhat\,\Delta t}
  \;=\;\prod_{|S|\le L}
        \exp\!\Bigl[-i\,\hat{f}_{m}(S)
          \Bigl(\textstyle\prod_{i\in S}\hat{n}_{i}\Bigr)
          \otimes\Sigma_{m}^{+}\otimes\qhat\,\Delta t\Bigr].
  \label{eq:product_synthesis}
\end{equation}
A Clifford $B_{m}$ on the suffix maps $\Sigma_{m}^{+}$ to
$\tilde{\Sigma}_{m}=\mathrm{diag}(+1,-1,0,\dots,0)$, any commuting Pauli set
being simultaneously diagonalisable by a
Clifford~\cite{Nielsen2010,Aaronson2004} compiled in
$\Order(n^{2}/\log n)$ gates by the Aaronson--Gottesman tableau algorithm.
In the rotated frame each factor of \cref{eq:product_synthesis} is a
multi-controlled momentum displacement: it displaces the qumode momentum by
$\mp\hat{f}_{m}(S)\Delta t$, conditioned on the prefix qubits in $S$ being
$|1\rangle$ and on the suffix lying in one of the two active states.  Since
$e^{-i\lambda\qhat}$ multiplies $|\eta\rangle$ by $e^{-i\lambda\eta}$ and
shifts $|\xi\rangle\mapsto|\xi-\lambda\rangle$, and displacements by
different amounts commute, the $\Order(n^{L})$ factors may be applied in any
order; such displacements are native on trapped-ion, dispersive circuit-QED,
and photonic platforms.

Two conventions fix the gate count: a controlled momentum displacement
$e^{-i\lambda\Pi\otimes\qhat}$ with $\Pi$ a basis projector counts as one
elementary gate, on the same footing as a controlled-phase rotation, and a
$k$-controlled gate compiles into $\Order(k)$ two-qubit gates through a
clean-ancilla ladder.  Each factor then carries at most $|S|+2m\le L+2m$
controls, $\le L$ prefix controls plus the $m$-fold selection of each active
suffix state, costing $\Order(L+m)=\Order(n)$ gates, so a family costs
$\Order(n^{L+1})$ gates beyond its Clifford; a native controlled
displacement onto the two-dimensional suffix subspace would remove the
$\Order(m)$ overhead.  On a rotated suffix eigenstate $|s\rangle$ the
circuit enacts
$|s\rangle|\eta\rangle\mapsto
e^{-i\lambda_{\mathcal{C}}(s)\eta\Delta t}|s\rangle|\eta\rangle$ with
$\lambda_{\mathcal{C}}(s)$ the family eigenvalue, so the product over
families reproduces $e^{-i\mathcal{H}(\eta)\Delta t}$ at every $\eta$
simultaneously.  \Cref{eq:product_synthesis} is verified operator-wise to $10^{-13}$ below.

\paragraph{State preparation}

The Lorentzian \cref{eq:qumode_initial} is not Gaussian, and a density
argument gives no rate.  We price it by the squeezed-Fock synthesis of Das
et al.~\cite{DasEtAl2026}, who treat exactly this class of hybrid kernel
states for the continuous-variable LCHS representation: the prepared state
\begin{equation}
  |\psi_{N_{F}}\rangle \;=\; S(r)\sum_{\nu=0}^{N_{F}-1}c_{\nu}\,|\nu\rangle,
  \label{eq:fock_kernel}
\end{equation}
synthesised by Law--Eberly or SNAP-plus-displacement protocols, has $L^{2}$
truncation error superalgebraic, $\Order(N_{F}^{-s/2})$ for every fixed $s$,
for Schwartz-class kernels, and subexponential,
$\Order(e^{-c\sqrt{N_{F}}})$, under strip analyticity
\cite[Thm.~2]{DasEtAl2026}.  Applied to the Lorentzian the rate is limited
by the tail: $|\psi|^{2}\sim\eta^{-4}$, the Fourier image of the kink of
$e^{-|\xi|}$ at $\xi=0$, is not Schwartz, only $s=1$ applies, and the
guaranteed decay is $\Order(N_{F}^{-1/2})$; we observe $\approx N_{F}^{-0.8}$
($1.7\times10^{-2}$ at $N_{F}=40$, optimised squeezing).

The kink is removable, an idea due to Jin, Liu, and
Ma~\cite{JinLiuMa2025,JinLiuMaIllPosed2025}, who smooth the initialisation
of the Schr\"{o}dingerised system to raise the approximation order in the
extended variable and reduce the extra register to about
$\log\log(1/\epsilon)$ qubits under a spectral discretisation of $\xi$.  The
mechanism transfers unchanged; only the notion of approximation differs,
ours being Fock truncation rather than a spectral grid.  The licence to
modify the profile is the domain-of-dependence argument of
\cref{sec:recovery_domain}: the warped-phase identity fixes the profile to
$e^{-\xi}$ only on the causal region, so recovery
\cref{eq:shifted_recovery} at $\xi^{*}>\lambda_{\max}t$ is unchanged by any
modification on $\xi<0$.  Replacing $e^{-|\xi|}$ by
$h(\xi)=e^{-\xi}\chi(\xi)$, with $\chi$ smooth, $1$ on $\xi\ge-1$ and $0$ on
$\xi\le-2$, makes $\hat{h}$ Schwartz, so \cite[Thm.~2(i)]{DasEtAl2026}
applies at every order and the truncation error decays faster than any
polynomial in $N_{F}$.  Full analyticity is unavailable: a kernel equal to
$e^{-\xi}$ on an interval and analytic would equal it everywhere.
Numerically the smooth kernel reaches $2.0\times10^{-2}$ at $N_{F}=20$ and
$2.4\times10^{-3}$ at $N_{F}=60$, against $3.1\times10^{-2}$ and
$1.3\times10^{-2}$ for the Lorentzian at matched squeezing; widening the
cutoff window, allowed whenever $\xi^{*}-\lambda_{\max}t$ is not small,
improves the constants at an $e^{-\xi^{*}}$ amplitude cost.  Since
\cref{eq:fock_kernel} is a finite Fock superposition every moment of $\qhat$
is finite, with $\|\Pi_{N_{F}}\qhat\Pi_{N_{F}}\|=\Order(\sqrt{N_{F}})$, the
property used in the Trotter bound; its non-Gaussianity is the
stellar rank $N_{F}-1$, and an $\epsilon$-accurate preparation perturbs the
post-selection probability by $\Order(\epsilon)$
\cite[Thm.~4]{DasEtAl2026}, so kernel error enters the budget linearly.

The other unpriced input is the amplitude encoding of
$\boldsymbol{\rho}(0)$.  Generic preparation of an arbitrary $N$-amplitude
state costs $\Order(N)$ gates and would defeat the polylogarithmic per-step
count, so what matters is the profiles actually used.  Xie and
Ben-Ami~\cite{XieBenAmi2025} give an efficient construction for discretised
Gaussians, and for the wider class of real polynomial functions
Gonzalez-Conde et al.~\cite{GonzalezConde2024} give two amplitude-encoding
methods.  Which is appropriate depends on the initial data, so we fold
neither into \cref{thm:endtoend}; the encoding is not an unknown black
box.  Finally, the Lorentzian may instead be carried
on a qubit register, the comparison of \cref{fig:cv}: truncating $\eta$ to
$[-\eta_{\max},\eta_{\max}]$ at spacing $\Delta\eta$ needs
$n_{\eta}=\lceil\log_{2}(2\eta_{\max}/\Delta\eta)\rceil$ qubits, as in the
QFT register of~\cite{TennieMagri2024}, with truncation error
$\Order(e^{-\eta_{\max}})$ and aliasing error $\Order(\Delta\eta)$, the slow
Lorentzian tail forcing $\eta_{\max}$ moderately large.

\begin{algorithm}
  \caption{Quantum simulation of \cref{eq:ode} via Fokker--Planck
           and CV-LCU Schr\"{o}dingerisation.}
  \label{alg:full}
  \begin{algorithmic}[1]
    \REQUIRE $\boldsymbol{\rho}(0)$, drift $\mathbf{f}$, diffusion $\sigma$,
             horizon $T$, Trotter step count $R$ meeting the target
             accuracy, grid size $N=2^{n}$.
    \ENSURE $|\boldsymbol{\rho}(T)\rangle$.
    \STATE Set $\Delta t:=T/R$.
    \STATE \textbf{Preprocessing.}  From the coefficients of $\mathbf{f}$,
           compute per axis and carry level $m$ the Walsh coefficients
           $\hat{f}_{m}(S)$, $|S|\le L$, of \cref{eq:Dm_monomials} and the
           boundary corrections of \cref{eq:A_split}; record the Cliffords
           $B_{m}$.
    \STATE \textbf{State preparation.}  Prepare
           $|\boldsymbol{\rho}(0)\rangle$~\cite{XieBenAmi2025,GonzalezConde2024}
           and the qumode kernel state of \cref{sec:qumode_prep}.
    \FOR{$r=1$ \TO $R$}
      \FOR{each family $S_{m,-}$ of $H_{2}$, then each $S_{m,+}$ of $H_{1}$}
        \STATE Apply $B_{m}$; apply the product
               \cref{eq:product_synthesis} of monomial-controlled phase
               rotations (Block~1) or momentum displacements (Block~2);
               apply $B_{m}^{\dagger}$.
      \ENDFOR
      \STATE Apply the boundary displacements of \cref{rem:boundary}.
    \ENDFOR
    \STATE \textbf{Post-processing.}  Inverse-QFT the qumode
           ($\eta\to\xi$); project onto $\xi>\lambda_{\max}(H_{1})T$
           (\cref{eq:shifted_recovery}; post-select or amplify); read out
           the qubit register.
  \end{algorithmic}
\end{algorithm}

\section{Complexity analysis}\label{sec:complexity}\label{sec:advantage}\label{sec:endtoend}

Throughout, $d$ is the spatial dimension, $N=2^{n}$ the grid points per
axis, $n_{\mathrm{tot}}=\sum_{k}n_{k}$ the total qubit count and
$N_{\mathrm{tot}}=2^{n_{\mathrm{tot}}}$ the total grid points; on the
isotropic grid $n_{k}=n$, so $n_{\mathrm{tot}}=dn$ and
$N_{\mathrm{tot}}=N^{d}$.  Further, $L$ is the drift degree, $T$ the
horizon, $\Delta t$ the Trotter step, and $N_{F}$ the Fock cutoff of
\cref{sec:qumode_prep}.  We keep $d$ explicit, the compression we claim
being a statement about it.

\paragraph{Gate count per Trotter step}

Fix an axis $k$ and a carry level $m\le n_{k}$.  By \cref{sec:synthesis} the
family exponential is an exact product of $\Order(n_{\mathrm{tot}}^{L})$
monomial-controlled displacements, each carrying at most
$L+2m$ controls and so costing $\Order(L+n_{k})$ gates under the
ancilla-ladder convention, with a suffix Clifford of
$\Order(n_{k}^{2}/\log n_{k})$~\cite{Aaronson2004}.  Three facts fix the
dimensional scaling: the families of distinct axes are disjoint
(\cref{thm:chromatic}), so they number $2n_{\mathrm{tot}}+d$, linear in the
total qubit count rather than $d$ times it; the suffix machinery, Clifford
and selection controls alike, never leaves one axis register, so it scales
with $n_{k}$; and only the prefix monomial count sees the full register, a
coupled drift of joint degree $L$ spreading its Walsh support over all
$n_{\mathrm{tot}}-m$ non-suffix qubits (\cref{rem:multidim_pauli}).  Summing
over families, parities, and the $\Order(1)$ boundary displacements of
\cref{rem:boundary},
\begin{equation}
  \text{gates per Trotter step}
  \;=\;\Order\!\Bigl(
      n_{\mathrm{tot}}^{L+1}\bigl(L+\textstyle\max_{k}n_{k}\bigr)
      \;+\;n_{\mathrm{tot}}\,\textstyle\max_{k}n_{k}^{2}\Bigr),
  \label{eq:gatespertrotter_gen}
\end{equation}
which on the isotropic grid is
\begin{equation}
  \boxed{\;
    \text{gates per Trotter step}
    \;=\;\Order\!\bigl(d^{L+1}n^{L+2}+d\,n^{3}\bigr)
    \;=\;\Order\!\bigl(d^{L+1}n^{L+2}\bigr)
    \quad (L\ge1),
  \;}
  \label{eq:gatespertrotter}
\end{equation}
equivalently $\Order(n_{\mathrm{tot}}^{L+2}/d)$, and polylogarithmic in
$N_{\mathrm{tot}}=2^{dn}$ at fixed degree and dimension.  A native
controlled displacement onto the two-dimensional suffix subspace would
remove the $\Order(L+n_{k})$ overhead and lower \cref{eq:gatespertrotter} to
$\Order(d^{L+1}n^{L+1})$.  The classical preprocessing is the closed-form
evaluation of the $\Order(n_{\mathrm{tot}}^{L})$ Walsh coefficients per
family, $\Order(d^{L+1}n^{L+1})$ arithmetic operations in total; no
$\Order(N_{\mathrm{tot}})$ assembly of $A$ is performed.

\paragraph{Number of Trotter steps}

For \cref{eq:trotter1} to reach global error $\epsilon$ the step count is
governed by $[H_{2}\otimes\Id,H_{1}\otimes\qhat]=[H_{2},H_{1}]\otimes\qhat$.
On the range of the kernel-state projector
$\|\qhat\|_{N_{F}}=\Order(\sqrt{N_{F}})$, so
with $\|H_{1}\|,\|H_{2}\|\le\|A\|$~\cite{Childs2021},
\begin{equation}
  R \;=\;\frac{T^{2}\,\|[H_{1},H_{2}]\|\,\|\qhat\|_{N_{F}}}{\epsilon}
   \;\le\;\frac{2T^{2}\|H_{1}\|\,\|H_{2}\|\,\|\qhat\|_{N_{F}}}{\epsilon}
   \;=\;\Order\!\left(
     \frac{T^{2}\|A\|^{2}\|\qhat\|_{N_{F}}}{\epsilon}\right),
  \label{eq:trottersteps}
\end{equation}
the rigorous state-dependent statement, including higher-order product
formulas, being \cite[Thm.~3]{DasEtAl2026}; a $2p$-th order
Suzuki--Trotter formula improves this to
$R=\Order\bigl(T^{1+1/(2p)}(\|A\|^{2}\|\qhat\|_{N_{F}})^{1/(2p)}
/\epsilon^{1/(2p)}\bigr)$.  The step count carries its own dimensional
factor through the stencil norm: $A$ is a sum of $d$ axis terms
\cref{eq:A_decomp} each of norm $\Order(\sigma/\Delta x^{2})$, so
$\|A\|=\Order(d\sigma/\Delta x^{2})$ and $R$ grows as $d^{2}$ at fixed grid.
This is the same $d^{2}$ a classical explicit solver pays per step, and is
not what the construction is meant to improve.

\paragraph{Total gate complexity}

Multiplying \cref{eq:gatespertrotter} by \cref{eq:trottersteps},
\begin{equation}
  \boxed{\;
    \text{total gates}
    \;=\;
    \Order\!\left(\frac{T^{2}\|A\|^{2}\,\|\qhat\|_{N_{F}}}{\epsilon}\,
                  d^{L+1}n^{L+2}\right),
    \qquad \|A\|=\Order\!\left(\frac{d\sigma}{\Delta x^{2}}\right).
  \;}
  \label{eq:totalcomplexity}
\end{equation}
The degree enters only through the per-step factor $d^{L+1}n^{L+2}$, the
price of the monomial synthesis, which is polylogarithmic in the grid size
whenever $n_{\mathrm{tot}}^{L+1}\ll2^{n_{\mathrm{tot}}}$, that is at any
fixed degree and moderate qubit count.  The two remaining factors are
the stencil norm and the post-selection overhead
$\Order\bigl(e^{\lambda_{\max}T}\|\boldsymbol{\rho}(0)\|/
\|\boldsymbol{\rho}(T)\|\bigr)$, which~\cite{AnLiuWangZhao2025} prove
unavoidable in the worst case for this operator class; both are examined
against the classical baseline in \cref{sec:advantage}.

\paragraph{Classical baseline and the regime of advantage}

The algorithm returns the density $\boldsymbol{\rho}(T)$ of \cref{eq:FP},
and its natural classical baseline is a grid Fokker--Planck
solver~\cite{Holubec2019} advancing the same discretised density under
$e^{A\Delta t}$.  That baseline carries an $\Order(N^{d})$ per-step
state-space cost, exponential in the dimension, and is \emph{stable in $T$}:
$A$ is a Markov generator with $\alpha(A)\le0$, $e^{At}$ is an $L^{1}$
contraction, and no exponential-in-$T$ growth appears.  Against it the
algorithm makes a definite trade, replacing the $\Order(N^{d})$ per-step
cost by the $\Order(d^{L+1}n^{L+2})=\Order\bigl(d^{L+1}(\log N)^{L+2}\bigr)$
gates of \cref{eq:gatespertrotter}, polynomial in $d$ and polylogarithmic in
the grid points per axis, while \emph{reintroducing} the
$e^{\lambda_{\max}T}$ post-selection factor that the classical density
solver does not carry.  The advantage we claim is this per-step compression
of the state space, not the horizon dependence.

This claim is distinct from the more familiar one that a long-time
\emph{trajectory} is hard.  Under the standard a priori analysis an explicit
order-$p$ integrator does carry a step count growing as $e^{\mu_{2}(J)T/p}$,
but the sharp constant is the one-sided Lipschitz constant
$\mu_{2}(J)=\lambda_{\max}\bigl((J+J^{\top})/2\bigr)$, the logarithmic norm
of the Jacobian $J=D\mathbf{f}$, which is $\le0$ for a dissipative flow, and
B-stable implicit methods remove the exponential
altogether~\cite{HairerWanner1996}.  Our post-selection factor
$e^{\lambda_{\max}(H_{1})T}$ is likewise exponential in $T$, with
$\lambda_{\max}(H_{1})$ the logarithmic norm $\mu_{2}(A)$ of the
\emph{generator} (\cref{sec:recovery_domain}).  The two rate constants are
governed by opposite signs of the divergence
$\nabla\!\cdot\!\mathbf{f}=\operatorname{tr}J$: in one dimension the
classical constant is $\max_{x}f'(x)/p$, largest where the flow expands,
while ours is $\approx-\tfrac12\min_{x}f'(x)$, largest where the flow
compresses, near the attractors where classical integrators are best
conditioned.  The classical rate also carries the $1/p$ that a higher-order
method shrinks, whereas $e^{\lambda_{\max}T}$ admits no such knob.  An
exponential-in-$T$ classical bound therefore does not, on its own, separate
the two routes; the separation we claim is the per-step one above.

Stiffness sharpens the same conclusion.  The positivity condition
$\Delta x\le2\sigma/F$, with $F:=\max_{j}|f_{j}|$, that secures
$\alpha(A)\le0$ fixes the coarsest admissible grid at $\Delta x=2\sigma/F$,
where
\begin{equation}
  \|A\|\;\approx\;\frac{2\sigma}{\Delta x^{2}}\;=\;\frac{F^{2}}{2\sigma},
  \qquad
  R\;\propto\;T^{2}\|A\|^{2}\;=\;\frac{T^{2}F^{4}}{4\sigma^{2}}.
  \label{eq:stiffness}
\end{equation}
The step count grows as the fourth power of the stiffness ratio $F$, and $T$
must be long enough to resolve the slow manifold, so the product $TF$ is
large by construction, whereas a classical implicit solver steps on the slow
timescale and is largely indifferent to $F$.  The method is therefore least
competitive in the stiff, timescale-separated regime; we do not claim
advantage there, and a stiff benchmark on which it is competitive is left to
future work.

\paragraph{End-to-end complexity}

Two accuracy parameters must be kept distinct.  The \emph{algorithmic}
accuracy $\epsilon$ measures how well the output approximates
$\boldsymbol{\rho}(T)=e^{AT}\boldsymbol{\rho}(0)$ in relative $L^{2}$ over
the physical interior; the diffusion $\sigma$ is a separate \emph{modelling}
parameter, and the bias between $\boldsymbol{\rho}(T)$ and a deterministic
observable of \cref{eq:ode} is $\Order(\sigma)$ (\cref{rem:recovery}).  The
five contributions and their knobs are:

\begin{center}
\small
\begin{tabular}{@{}lll@{}}
  \toprule
  error source & knob & set to reach $\epsilon$ \\
  \midrule
  spatial discretisation, $\Order(\Delta x^{2})$ & $\Delta x$
    & $\Delta x=\Order(\sqrt{\epsilon})$, capped by $\Delta x\le2\sigma/F$ \\
  outer Trotter \cref{eq:trotter1} & $R$
    & $R=\Order(T^{2}\|A\|^{2}\|\qhat\|_{N_{F}}/\epsilon)$ \\
  inner Trotter \cref{eq:trotter_H2} & $R$
    & $\Order(n_{\mathrm{tot}}^{2})$ family pairs, folded into $R$ \\
  kernel truncation \cref{eq:fock_kernel} & $N_{F}$
    & $N_{F}=\epsilon^{-o(1)}$ or $\epsilon^{-2}$ (smoothed / Lorentzian) \\
  recovery amplitude \cref{eq:psucc} & AA rounds
    & $\Order(e^{\lambda_{\max}T}\|\boldsymbol{\rho}(0)\|/\|\boldsymbol{\rho}(T)\|)$ \\
  \bottomrule
\end{tabular}
\end{center}

\noindent
The scaling is fixed by the stencil norm being dominated by the diffusion
diagonal, $\|A\|\approx2\sigma/\Delta x^{2}$: resolving the density to
relative $L^{2}$ accuracy $\epsilon$ forces $\Delta x=\Order(\sqrt{\epsilon})$
at fixed $\sigma$, while a deterministic target forces
$\sigma=\Theta(\epsilon)$ at the positivity grid $\Delta x=2\sigma/F$, and
either way $\|A\|=\Theta(1/\epsilon)$.  This is a property of the explicit
stencil, not of the product formula.

\begin{theorem}[End-to-end cost]
  \label{thm:endtoend}
  Let $\mathbf{f}$ be a polynomial drift of degree $L$ on a box, discretised
  on the grid of \cref{lem:affine} with $N=2^{n}$ points and diffusion
  $\sigma>0$ satisfying $\Delta x\le2\sigma/F$, so $A$ is semi-stable.
  Suppose $\boldsymbol{\rho}_{\sigma}(\cdot,t)$ has bounded derivatives up to
  fourth order on $[0,T]$, and adopt \cite[Thms.~2--4]{DasEtAl2026}.  Then
  for any $\epsilon\in(0,1)$ in relative $L^{2}$ over the physical interior,
  \cref{alg:full} with the smoothed kernel at cutoff $N_{F}(\epsilon)$,
  first-order Trotterisation with $R(\epsilon)$ steps, and terminal
  amplitude amplification returns a state $\epsilon$-close to
  $\boldsymbol{\rho}(T)$ using
  \begin{equation}
    \boxed{\;
    N_{\mathrm{gate}}
    \;=\;\widetilde{\Order}\!\left(
      \frac{T^{2}\,\|A\|^{2}\,\sqrt{N_{F}}}{\epsilon}\;d^{L+1}n^{L+2}
    \right)\times
    \Order\!\left(e^{\lambda_{\max}T}\,
      \frac{\|\boldsymbol{\rho}(0)\|}{\|\boldsymbol{\rho}(T)\|}\right)
    \;}
    \label{eq:endtoend}
  \end{equation}
  gate applications, where $\widetilde{\Order}$ hides factors
  polylogarithmic in $N_{\mathrm{tot}}$ and $1/\epsilon$, including the
  $n_{\mathrm{tot}}^{2}$ inner-family factor of \cref{eq:trotter_H2}, and the
  smoothed kernel gives $\sqrt{N_{F}}=\epsilon^{-o(1)}$, against
  $\Theta(\epsilon^{-1})$ for the raw Lorentzian.
\end{theorem}

\begin{proof}
  Set each row of the budget table to contribute $\Order(\epsilon)$.  The
  positivity condition caps $\Delta x$, whose $\Order(\Delta x^{2})$ error is
  then $\Order(\epsilon)$ under the fourth-derivative bound.  The outer and
  inner Trotter errors are both $\Order(\Delta t^{2}\Gamma)$ per step with
  $\Gamma=\Order(\|A\|^{2}(\|\qhat\|_{N_{F}}+n_{\mathrm{tot}}^{2}))$, since
  $\|[H_{1},H_{2}]\|\le2\|A\|^{2}$ and each of the
  $\Order(n_{\mathrm{tot}}^{2})$ family commutators obeys
  $\|[H_{S_{m}},H_{S_{m'}}]\|=\Order(\|A\|^{2})$; setting the accumulated
  $\Order(T^{2}\Gamma/R)$ to $\epsilon$ gives
  $R=\widetilde{\Order}(T^{2}\|A\|^{2}\sqrt{N_{F}}/\epsilon)$.  The cutoff is
  fixed by \cite[Thm.~2]{DasEtAl2026} and the errors compose linearly into
  the post-selection probability by \cite[Thm.~4]{DasEtAl2026}.  Multiplying
  $R$ by \cref{eq:gatespertrotter} and by the \cref{eq:aa_cost} rounds gives
  \cref{eq:endtoend}.
\end{proof}

\begin{corollary}[Cost for a deterministic observable]
  \label{cor:deterministic}
  If the target is an observable of \cref{eq:ode} to accuracy $\epsilon$,
  the small-noise bias forces $\sigma=\Theta(\epsilon)$; at the positivity
  grid $\|A\|=\Theta(dF^{2}/\epsilon)$ and $n=\Order(\log(F/\epsilon))$, so
  \cref{eq:endtoend} becomes
  \begin{equation}
    N_{\mathrm{gate}}
    \;=\;\widetilde{\Order}\!\big(T^{2}\,d^{L+3}\,F^{4}\,\epsilon^{-3}\big)
      \times e^{\lambda_{\max}T}
    \label{eq:eps3}
  \end{equation}
  with the smoothed kernel, degrading to $\epsilon^{-4}$ with the
  Lorentzian, and polylogarithmic in all remaining parameters.
\end{corollary}

The dimension enters \cref{eq:eps3} polynomially at fixed $L$, where the
classical grid solver pays $N^{d}=\epsilon^{-\Theta(d)}$; that gap is the
claim, and it holds at fixed $L$ only, the factor $d^{L+3}$ degrading as the
degree grows.  The $\epsilon^{-3}$ is polynomially worse than the
$\operatorname{polylog}(1/\epsilon)$ of high-precision linear
solvers~\cite{Childs2021,Krovi2023,AnChildsLin2023}, for two reasons: the
diffusion-limited $\|A\|=\Theta(1/\epsilon)$, intrinsic to the explicit
stencil, and the $1/\epsilon$ of first-order Trotterisation.  A $2p$-th
order formula~\cite{Tranter2019} improves the explicit $\epsilon$ and $T$
powers but leaves the $\|A\|$ coupling untouched, so the accuracy dependence
stays polynomial.  What the method buys in exchange is the oracle-free,
closed-form compilation of \cref{sec:implementation}.

\section{Numerical validation}\label{sec:numerics}\label{sec:num_fp}\label{sec:num_pauli}\label{sec:num_cv}\label{sec:num_e2e}

We validate by exact classical simulation of the underlying mathematics: the
qumode is represented on a fine $\eta$ grid, the standard classical proxy for
a continuous-variable mode, and each $e^{-i\mathcal{H}(\eta)t}$ is applied
exactly, isolating the algebra of the algorithm from Trotter and hardware
error.  The $\eta$ grid must resolve a phase winding at rate
$2\pi\|A\|T$, so its cost grows with the stencil norm; at $n=6$ we take
$\eta\in[-10,10]$ with $1.6\times10^{4}$ points.  This is a limitation of
the proxy rather than of the algorithm, and it bites hardest in exactly the
regime of interest.

The scope is narrow.  These runs test the structural results, the
product synthesis, the shifted recovery, the end-to-end accuracy, and the
qumode-versus-register comparison of \cref{sec:cv_lcu}.  They do not measure
the gate counts of \cref{sec:complexity}, as the mode Hamiltonians are applied
exactly rather than compiled. They do not simulate the split \cref{eq:trotter1},
whose error is left to \cref{eq:trottersteps}, so the accuracies reported are
those of the mathematics only.

The benchmark is the bistable gradient flow $\dot x=x-x^{3}$ (degree $L=3$)
on $[-0.5,1.5]$ with $\sigma=0.08$ and $T=2$, on $n=6$ qubits ($N=64$), from
a discretised Gaussian $\boldsymbol{\rho}(0)$ centred at $x_{0}=0.6$ of width
$w_{0}=0.12$.  Both $\sigma$ and $w_{0}$ are constrained from two sides:
small enough to keep the density concentrated, so that the small-noise
picture of \cref{sec:num_fp} applies, and large enough for the positivity
condition $\Delta x\le2\sigma/\max_{j}|f_{j}|$ of \cref{sec:cd}
($0.032\le0.085$) and for $\boldsymbol{\rho}(0)$ to be resolved, the chosen
width spanning about four cells per standard deviation.  All draws use a fixed seed, and the
runs are on the CPU of an M3 MacBook Air.  Where a second drift sharpens a point we quote
a comparison run on the logistic law $\dot x=a\,x(1-x)$, $a=1.5$ (degree
$L=2$) on $[-0.2,1.2]$ with $\sigma=0.01$, $T=1.5$, whose behaviour is
identical in kind.

\paragraph{Fokker--Planck linearisation}

With the conservative stencil of \cref{sec:cd} every column of $A$ sums to
zero to machine precision, so the discrete evolution preserves probability
exactly; the naive stencil, leaving the boundary diagonals at
$-2\sigma/\Delta x^{2}$, leaks at both walls, its column sums departing from
zero by $\Order(10)$ against $7.1\times10^{-15}$ for the conservative form,
and its mass falling to $0.966$ by $t=T$.  An
independent Euler--Maruyama integration of \cref{eq:sde} reproduces
$e^{At}\boldsymbol{\rho}_{0}$ (\cref{fig:fp}(a)), confirming $A$ as the
correct generator, the grid-free mean of $x$ agreeing to a relative
$1.0\times10^{-3}$.

The trajectory is read back out here.  With weak diffusion and concentrated
$\boldsymbol{\rho}(0)$ the density peak tracks the deterministic solution and
settles on the fixed point $x=1$ (\cref{fig:fp}(b)), reaching $0.988$
against the deterministic $x(T)=0.984$ and agreeing to $0.044$ in $x$ over
the whole interval, or $2\%$ of the domain width; peaks are located by parabolic
fit through the maximal cell and its two neighbours.

\begin{figure}[tbp]
  \centering
  \includegraphics[width=0.88\textwidth]{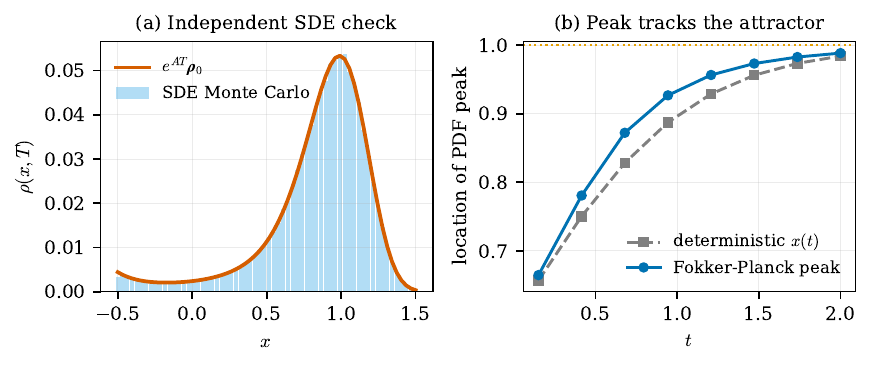}
  \caption{Fokker--Planck linearisation ($\dot x=x-x^{3}$).
    (a)~An independent Monte-Carlo integration of the stochastic
    differential equation matches $e^{AT}\boldsymbol{\rho}_{0}$.
    (b)~The density peak, located by parabolic sub-cell interpolation,
    follows the deterministic trajectory onto the attractor to within
    $0.044$ in $x$.}
  \label{fig:fp}
\end{figure}

\paragraph{Structural results}

The exhaustive Pauli decompositions of $H_{1}$ and $H_{2}$ confirm the
structural results of \cref{sec:pauli}.  Every non-zero string has the
bipartite $\{I,Z\}$-prefix\,/\,$\{X,Y\}$-suffix form of
\cref{thm:pauli_structure}; the strings partition into exactly $2n+1=13$
commuting families, $H_{1}$ populating the even-$Y$ families and $H_{2}$ the
odd-$Y$ ones as \cref{thm:chromatic} requires.  The per-family term counts
respect \cref{eq:term_count}, totalling $347$ for the generic degree-three
drift, whereas the logistic comparison run, on a grid centred on that
drift's axis of symmetry $x=1/2$, annihilates further Walsh coefficients and
falls strictly sparser at $226$.  Summed over families
the term count grows as $\Theta(2^{n})=\Theta(N)$, the suffix
multiplicities restoring an exponential scaling in the string count, which
is why \cref{sec:synthesis} pays the $\Order(n^{L})$ prefix monomial count
instead.  Finally, \cref{eq:product_synthesis} is verified operator-wise:
the product of multi-controlled displacements reproduces
$e^{-iH_{S_{m,+}}\otimes\qhat\Delta t}$ to better than $10^{-13}$, at
$\Order(n^{L})$ commuting factors.

\paragraph{Continuous-variable versus discretised LCU}

A digital alternative encodes $\eta$ on an $n_{\eta}=\log_{2}M$-qubit
register of range $\eta_{\max}$ and spacing $\Delta\eta=2\eta_{\max}/(M-1)$,
reproducing the qumode only as $\eta_{\max}\to\infty$, $\Delta\eta\to0$.
\Cref{fig:cv} measures its recovery error against the continuum qumode under
refinement: each curve fixes $\eta_{\max}$ and adds qubits, the error
falling until it meets a \emph{truncation floor} set by $\eta_{\max}$, below
which resolution buys nothing, while a larger range lowers the floor only at
the cost of more qubits to reach it, the slowly-decaying Lorentzian tail
forcing a wide range.  The qumode sits at the joint limit with no register
at all, attaining an accuracy the digital encoding reaches only by enlarging
range and resolution together.  This is an accuracy-per-resource advantage
of the representation, not an asymptotic gate-count separation, and it is
the numerical content of the qumode--register comparison.

\begin{figure}[tbp]
  \centering
  \includegraphics[width=0.56\textwidth]{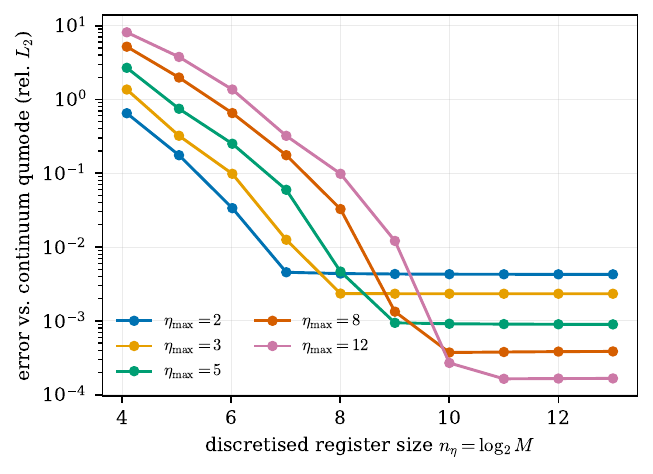}
  \caption{A discretised (qubit-register) LCU approaching the
    continuous-variable qumode limit.  Each curve fixes the register range
    $\eta_{\max}$ and refines its resolution; the error plateaus at a
    truncation floor set by $\eta_{\max}$.  The qumode realises the joint
    limit $\eta_{\max}\to\infty$, $\Delta\eta\to0$ with no register.}
  \label{fig:cv}
\end{figure}

\paragraph{End-to-end accuracy}

The full pipeline returns $\boldsymbol{\rho}(T)$ to a relative interior error
of $2.5\times10^{-4}$ against $e^{AT}\boldsymbol{\rho}_{0}$
(\cref{fig:e2e}(a)).  Both abscissas were computed
directly: $\lambda_{\max}(H_{1})=+0.725$, comfortably inside
\cref{thm:abscissa}, whose interior term alone is $2.88$, while $\alpha(A)$
vanishes to machine precision,
$|\alpha(A)|<10^{-13}$.  The generator is therefore semi-stable with
indefinite Hermitian part and the $L^{2}$ norm grows transiently,
$\max_{t\le T}\|e^{At}\|_{2}=1.28$.  The recovery error against $\xi^{*}$
falls sharply at $\xi^{*}=\lambda_{\max}T$ (\cref{fig:e2e}(c)): at
$\lambda_{\max}T=1.45$ it is $0.59$ at $\xi^{*}=1.0$ and $0.15$ at
$\xi^{*}=1.4$, then $5.0\times10^{-4}$ at $\xi^{*}=1.6$ and below $10^{-3}$
thereafter, the logistic run behaving the same way about its own
$\lambda_{\max}T=2.41$.  Recovery below $\lambda_{\max}t$ fails and above it
succeeds, as \cref{sec:recovery_domain} requires, so we take
\begin{equation}
  \xi^{*}(t)\;=\;\lambda_{\max}(H_{1})\,t+\tfrac12
  \label{eq:xistar_rule}
\end{equation}
throughout: this needs no reference solution, $\lambda_{\max}(H_{1})$ being
read off the discretisation, and sits clear of the certified boundary at a
cost of $e^{-1/2}$ in amplitude, giving $\xi^{*}=1.95$ here.  The error
stays below $10^{-3}$ over the whole interval $t\in(0,T]$
(\cref{fig:e2e}(b)), the mild growth at the largest $t$ tracing to the
boundary cell where the positivity condition is marginal; the physical
interior is unaffected, and the recovered peak coincides with the
deterministic attractor throughout.

\begin{figure}[tbp]
  \centering
  \includegraphics[width=\textwidth]{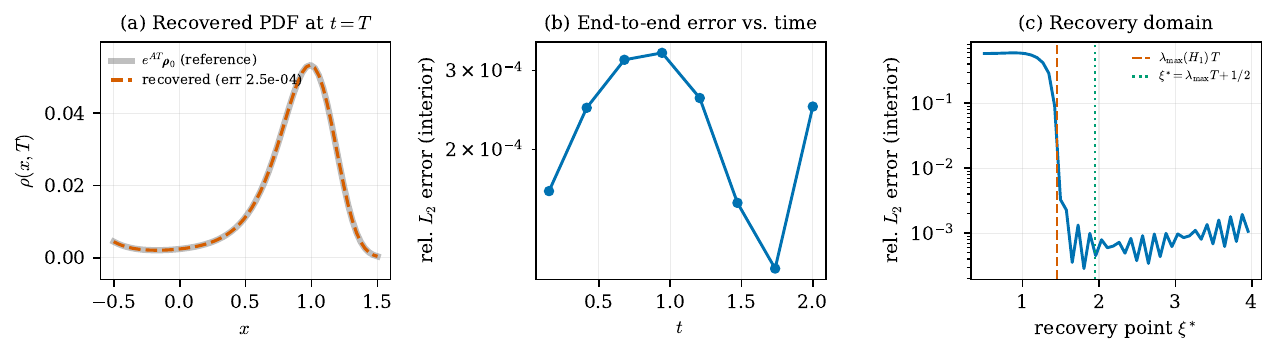}
  \caption{End-to-end validation ($\dot x=x-x^{3}$).
    (a)~The recovered density matches the reference $e^{AT}\boldsymbol{\rho}_{0}$.
    (b)~The relative interior error stays below $10^{-3}$
    as a function of the simulation time $t$.
    (c)~Recovery error against the recovery point $\xi^{*}$, with the
    certified onset $\xi^{*}=\lambda_{\max}(H_{1})T$ marked; the error
    falls by three orders of magnitude there.}
  \label{fig:e2e}
\end{figure}

\section{Discussion}\label{sec:discussion}\label{sec:trotter_qumode}

\paragraph{Advantages of the CV-LCU approach}

The continuous-variable coupling offers three benefits over fully digital
Schr\"{o}dingerisation~\cite{TennieMagri2024,JinLiu2022,JinLiuYuFP2024}: the
$\eta$-spectrum is represented natively, where a register introduces the
aliasing and truncation of \cref{fig:cv}; many platforms support
$H_{1}\otimes\qhat$ as an always-on Jaynes--Cummings or dispersive
interaction~\cite{JinLiuJC2024}; and the whole continuum evolves in one
circuit, removing the per-mode sampling overhead.  These concern the
register, not the mode's own resources.  The qumode's resource is the
prepared kernel state, whose cutoff $N_{F}$, equivalently stellar rank
$N_{F}-1$~\cite{DasEtAl2026}, controls the kernel accuracy and whose
truncated norm $\|\qhat\|_{N_{F}}=\Order(\sqrt{N_{F}})$ enters the Trotter
count.  What the continuum removes is the register: the digital alternative
needs $n_{\eta}$ ancilla qubits, QFT and controlled-phase machinery on them,
and the range--resolution trade-off of \cref{fig:cv}, where the qumode is a
single native mode with no ancilla; the same elimination is explicit in the
hybrid LCHS setting of~\cite{DasEtAl2026}.  In the discrete analysis
of~\cite{JinLiu2022} the mode register contributes
$\|\boldsymbol{D}\|_{\max}=\Order(1/\epsilon)$ to
$\|H_{\mathrm{total}}\|_{\max}$; that factor does not vanish in the
continuum limit but reappears as $\|\qhat\|_{N_{F}}$, paid in Trotter steps
on one mode rather than in qubits, phases, and QFT depth on a register.

\paragraph{Trotter error and higher-order schemes}

The Trotter error for \cref{eq:Htotal} cannot be bounded through operator
norms alone: for an ideal continuum mode $\|\qhat\|=\infty$, and for the
ideal Lorentzian $\langle\qhat^{2}\rangle$ converges while
$\langle\qhat^{4}\rangle$ diverges, so any bound needing high moments of the
ideal state fails.  A state-dependent bound is supplied by
\cite[Thm.~3]{DasEtAl2026}, which treats exactly the class
$\qhat\otimes L+\Id\otimes H$: the prepared state \cref{eq:fock_kernel} lies
in the range of $\Pi_{N_{F}}$, the evolution may be analysed with
$\|\qhat\|_{N_{F}}=\|\Pi_{N_{F}}\qhat\Pi_{N_{F}}\|=\Order(\sqrt{N_{F}})$,
and a $p$th-order formula needs
$n_{t}=\Order(t^{1+1/p}(\Gamma_{p,N_{F}}/\epsilon_{t})^{1/p})$ steps, with
$\Gamma_{p,N_{F}}$ collecting nested commutator sums weighted by powers of
$\|\qhat\|_{N_{F}}$.  For \cref{eq:trotter1} the leading error is
$\Delta t^{2}\|[H_{1},H_{2}]\|\cdot\|\qhat\|_{N_{F}}$, so
\cref{eq:trottersteps} carries that extra factor.  Larger cutoffs therefore
buy kernel accuracy at the price of Trotter steps: the continuum-mode
counterpart of the range--resolution trade-off of a discretised register,
governed by the moments of the prepared state rather than by a register
range, with the two errors composing linearly into the post-selection
probability \cite[Thm.~4]{DasEtAl2026}.  Higher-order or optimally-ordered
formulas~\cite{Tranter2019} mitigate this and reduce depth.

\paragraph{Recovery and observable estimation}

After \cref{alg:full} the qumode encodes a superposition over $\xi$-space;
an inverse QFT maps $\eta\to\xi$, and projecting onto the shifted causal
region $\xi>\lambda_{\max}(H_{1})T$ of \cref{eq:shifted_recovery}, or
amplitude-amplifying, recovers $|\boldsymbol{\rho}(T)\rangle$.  The
Fokker--Planck semigroup contracts in $L^{1}$ but not in the $L^{2}$ norm
that amplitude encoding uses: $\|e^{At}\|_{2}\le e^{\lambda_{\max}(H_{1})t}$,
and transient $L^{2}$ growth is real for the non-normal generator
($\max_{t\le T}\|e^{At}\|_{2}=1.28$ on the benchmark of
\cref{sec:numerics}).  The branch carrying $\boldsymbol{\rho}(T)$ enters the
post-transform state, after the shifted projection, with amplitude
proportional to
$e^{-\lambda_{\max}T}\|\boldsymbol{\rho}(T)\|/\|\boldsymbol{\rho}(0)\|$, so
the projection succeeds with probability
\begin{equation}
  p_{\mathrm{succ}}
  \;=\;\Omega\!\left(
    e^{-2\lambda_{\max}T}\,
    \frac{\|\boldsymbol{\rho}(T)\|^{2}}{\|\boldsymbol{\rho}(0)\|^{2}}
  \right),\qquad \lambda_{\max}:=\max\bigl(\lambda_{\max}(H_{1}),0\bigr).
  \label{eq:psucc}
\end{equation}
Amplitude amplification~\cite{Brassard2002} raises this to $\Order(1)$ using
$\Order(1/\sqrt{p_{\mathrm{succ}}})$ rounds rather than the
$\Order(1/p_{\mathrm{succ}})$ of naive repetition, an overhead of
\begin{equation}
  \Order\!\left(\frac{1}{\sqrt{p_{\mathrm{succ}}}}\right)
  \;=\;
  \Order\!\left(
    e^{\lambda_{\max}T}\,
    \frac{\|\boldsymbol{\rho}(0)\|}{\|\boldsymbol{\rho}(T)\|}
  \right)
  \label{eq:aa_cost}
\end{equation}
applications of the circuit of \cref{alg:full} and its inverse; it needs
reflections about the initial state and about the $\xi>\lambda_{\max}T$
subspace, both available here since the state preparation is unitary and the
projector is a comparison on the $\xi$-register.  Without the quadratic
saving an exponentially decaying $\|\boldsymbol{\rho}(T)\|$ would put
recovery out of reach.  The bound \eqref{eq:aa_cost} is stated at the
theoretical shift $\lambda_{\max}T$, the smallest admissible recovery point;
the choice \cref{eq:xistar_rule} sits a fixed $\tfrac12$ above it, so the
realised amplitude penalty is $e^{-\xi^{*}}$, lowering the success
probability by $e^{-1}$ and raising the overhead by $e^{1/2}\approx1.65$,
uniformly in $t$.  The factor $e^{\lambda_{\max}T}$ is exponential in the
horizon, and is not an artifact of our construction: for coefficient
matrices whose eigenvalues have differing real parts, An, Liu, Wang, and
Zhao~\cite{AnLiuWangZhao2025} prove that some initial conditions force a
cost exponential in $T$ for \emph{any} quantum ODE solver, and Tennie and
Magri~\cite{TennieMagri2025} state the same consequence for exactly this
operator class.  What \cref{thm:abscissa} adds is the mechanism and the
rate: the exponent is $\lambda_{\max}(H_{1})$, bounded by the compression of
the flow in the interior and by the outward drift at the walls, and by
neither the grid spacing nor the diffusion diagonal $\sigma/\Delta x^{2}$,
so the overhead is $e^{\Order(T)}$ with a constant fixed by the vector field
and the domain, and refining the grid does not worsen it.

\paragraph{Observable estimation}
Physical observables of interest are low-dimensional functionals of the
density: means, variances, and probabilities of regions.  Amplitude encoding
requires care, since measuring a diagonal observable $O$ in
$|\boldsymbol{\rho}(t)\rangle$ returns the \emph{quadratic} functional
$\sum_{j}O_{j}\rho_{j}^{2}/\|\boldsymbol{\rho}\|^{2}$, not the linear
$\sum_{j}O_{j}\rho_{j}$.  The linear functional is obtained as an overlap:
with $|\varphi_{O}\rangle=\sum_{j}O_{j}|j\rangle/\|O\|_{2}$,
\begin{equation}
  \sum_{j}O_{j}\rho_{j}(t)
  \;=\;\|O\|_{2}\,\|\boldsymbol{\rho}(t)\|\,
       \langle\varphi_{O}|\boldsymbol{\rho}(t)\rangle,
  \label{eq:linear_readout}
\end{equation}
and a Hadamard test estimates the overlap to precision $\epsilon_{o}$ with
$\Order(1/\epsilon_{o}^{2})$ repetitions, or $\Order(1/\epsilon_{o})$ with
amplitude estimation, a sample cost
$\Order(\|O\|_{2}^{2}\|\boldsymbol{\rho}\|_{2}^{2}/\epsilon^{2})$ for
$\langle O\rangle$ to additive precision $\epsilon$.  For a \emph{dense}
observable this is expensive, $\|O\|_{2}^{2}=\Theta(N)$ giving
$\Order(N/(k\epsilon^{2}))$ samples for a density supported on $k$ cells;
the low-dimensional observables above are not of that type.  Restricting $O$
to the $\Order(k)$-cell region carrying the density, at a truncation error
controlled by the concentration of $\boldsymbol{\rho}$, gives
$\|O\|_{2}^{2}=\Order(k\|O\|_{\infty}^{2})$ while
$\|\boldsymbol{\rho}\|_{2}^{2}=\Order(1/k)$, hence
$\Order(\|O\|_{\infty}^{2}/\epsilon^{2})$ samples, independent of $N$.  The
bound assumes the $\Order(k)$-cell support is known; locating it is itself
part of the computation, and for a multistable system one does not know a
priori which basin the density settles into, but a coarse low-resolution
pre-pass identifies the occupied cells, after which the refined estimate
proceeds as above.  Reference states $|\varphi_{O}\rangle$ for low-degree
grid polynomials and region indicators are preparable in $\mathrm{poly}(n)$
gates; ratio functionals
$\langle g\rangle_{\rho}=\sum_{j}g_{j}\rho_{j}/\sum_{j}\rho_{j}$ follow from
two such overlaps and are independent of the normalisation; and
$\|\boldsymbol{\rho}(t)\|$ need not be tracked classically, since
$\|\boldsymbol{\rho}(0)\|$ is known from the prepared initial data and the
ratio $\|\boldsymbol{\rho}(T)\|/\|\boldsymbol{\rho}(0)\|$ is estimated from
the post-selection success frequency of \cref{eq:psucc}, or quadratically
faster by amplitude
estimation~\cite{Brassard2002,Crawford2021,Reggio2024}.

\section{Conclusion}\label{sec:conclusion}

We have presented a hybrid qubit--qumode algorithm for the nonlinear ODE
$\dot{\mathbf{x}}=\mathbf{f}(\mathbf{x})$ with polynomial drift of
degree~$L$, returning the trajectory and the low-order statistics of the
flow through the density that carries them, in four stages:
Fokker--Planck linearisation; central-difference discretisation to
$\dot{\boldsymbol{\rho}}=A\boldsymbol{\rho}$ with sparse tridiagonal $A$;
Schr\"{o}dingerisation of the non-unitary semigroup into a parametrised
family of unitary flows; and CV-LCU simulation, in which one qumode encodes
the entire Fourier-mode continuum and
$H_{1}\otimes\qhat+H_{2}\otimes\Id$ is Trotterised into family exponentials
synthesised exactly by monomial-controlled momentum displacements.

The bipartite Pauli structure of a tridiagonal generator, and its partition
into $\Order(\log N)$ commuting families, are due to Arseniev et
al.~\cite{Arseniev2024} (\cref{thm:pauli_structure}).  What a polynomial
drift adds is the content of \cref{thm:chromatic,thm:pauli_count}: the
Hermitian split sends $H_{1}$ into the even-$Y$ families and $H_{2}$ into
the odd-$Y$ ones, the count carries to $2n_{\mathrm{tot}}+d$ in $d$
dimensions, and each family factorises into a degree-$L$ diagonal and a
fixed rank-two bond operator carrying at most $\Order(2^{m}n^{L})$ terms on
$\Order(n^{L})$ prefix monomials.  The factorisation makes each family exponential
an exact product of monomial-controlled displacements
(\cref{sec:synthesis}), giving $\Order(d^{L+1}n^{L+2})$ gates per step and a
total first-order complexity
$\Order\bigl((T^{2}\|A\|^{2}\|\qhat\|_{N_{F}}/\epsilon)d^{L+1}n^{L+2}\bigr)$.
Recovery is performed on the
shifted half-line $\xi>\lambda_{\max}(H_{1})t$ dictated by the numerical
abscissa of the non-normal generator, at an $e^{-\lambda_{\max}T}$ amplitude
cost, and \cref{thm:abscissa} bounds that abscissa uniformly in the grid
spacing.  An exact classical simulation confirms the theorems, the synthesis
identity, the shifted recovery, and the end-to-end accuracy
(\cref{sec:numerics}).

Several directions remain, ordered by how much each would resolve.  The
numerical evidence should be extended in the two ways identified in
\cref{sec:numerics}: a Trotterised end-to-end run, closing the gap between
the accuracy of the mathematics reported here and that of the circuit of
\cref{alg:full} as specified, and a two-dimensional benchmark exhibiting the
per-step compression that \cref{sec:advantage} argues for and no
one-dimensional run can display.  Three theoretical questions follow.  Our
bounds assume an exactly polynomial drift, and propagating the approximation
error of a degree-$L$ approximant through the discretisation and the Walsh
coefficients would extend them to the smooth case.  And since the theorems of \cref{sec:pauli} concern only
the Hermitian split of $A$, they are inherited by any dilation coupling an
ancilla to that split; pairing them with the tight-binding ancilla
of~\cite{Li2026Dilation} is open.  Qumode error models covering photon loss,
thermal noise, and dephasing, with hardware-aware compilation of the
qubit--qumode coupling, are also needed before the counts of
\cref{sec:complexity} become fault-tolerant estimates.

\paragraph{Reducing the dependence on the grid}
The grid size does not enter where the circuit is built: the per-step count
\cref{eq:gatespertrotter} is already polylogarithmic in the grid points, and
the burden falls on the step count, $R\propto T^{2}\|A\|^{2}$ with
$\|A\|=\Theta(d\sigma/\Delta x^{2})$, so $R$ grows like $\Delta x^{-4}$.
Three routes could reduce it.  The bound \cref{eq:trottersteps} uses
$\|[H_{1},H_{2}]\|\le2\|A\|^{2}$, which discards the structure of the
stencil; sharper commutator-scaling bounds~\cite{ChildsSuTran2021} should improve the
scaling of $R$ from $\Theta(\Delta x^{-4})$ to a smaller power of $\Delta x$, since the discretised commutator approximates a
commutator of differential operators.  A $2q$-th order stencil would need only
$\Delta x=\Order(\epsilon^{1/2q})$ rather than $\Order(\sqrt{\epsilon})$, at
the price of a banded rather than tridiagonal operator, whose Pauli structure
is open~\cite{Arseniev2024}.  Neither touches the coupling
$\|A\|=\Theta(1/\epsilon)$, which is a stiffness of the \emph{explicit}
stencil: an implicit or exponential integrator would remove it but does not
sit naturally with Schr\"{o}dingerisation, which lifts the semigroup rather
than a time-stepped approximation of it.

\appendix

\section{Proof of \texorpdfstring{\cref{thm:abscissa}}{Theorem 6.1}}
\label{app:abscissa}

Since $\boldsymbol{\rho}$ is real,
$\boldsymbol{\rho}^{\!\top}\!A\boldsymbol{\rho}
=\boldsymbol{\rho}^{\!\top}\!H_{1}\boldsymbol{\rho}$, so it suffices to
bound the Rayleigh quotient of $A$.  Write
$a_{j}:=A_{j+1,j}=\sigma/\Delta x^{2}+f_{j}/(2\Delta x)$ and
$b_{j}:=A_{j,j+1}=\sigma/\Delta x^{2}-f_{j+1}/(2\Delta x)$, both non-negative
under $\Delta x\le2\sigma/F$.  The master-equation form sets
$A_{jj}=-\sum_{i\neq j}A_{ij}$, so every entry attaches to one bond and
\begin{equation}
  \boldsymbol{\rho}^{\!\top}\!A\boldsymbol{\rho}
  \;=\;\sum_{j=0}^{N-2}\Bigl[(a_{j}+b_{j})\rho_{j}\rho_{j+1}
     -a_{j}\rho_{j}^{2}-b_{j}\rho_{j+1}^{2}\Bigr];
  \label{eq:bondsum}
\end{equation}
with $a_{j}+b_{j}=2\alpha_{j}$, $\alpha_{j}-a_{j}=-\beta_{j}$,
$\alpha_{j}-b_{j}=\beta_{j}$ from \cref{eq:alpha_def,eq:beta_def}, and
$2\alpha_{j}\rho_{j}\rho_{j+1}=-\alpha_{j}(\rho_{j+1}-\rho_{j})^{2}
+\alpha_{j}(\rho_{j}^{2}+\rho_{j+1}^{2})$, this is \cref{eq:quadform}.
Summation by parts on $u_{j}:=\rho_{j}^{2}$, using
$\beta_{j}-\beta_{j-1}=(f_{j+1}-f_{j-1})/(4\Delta x)$, then gives
\begin{equation}
  \boldsymbol{\rho}^{\!\top}\!H_{1}\boldsymbol{\rho}
  \;=\;-D
   \;-\;\sum_{j=1}^{N-2}\frac{f_{j+1}-f_{j-1}}{4\Delta x}\,\rho_{j}^{2}
   \;+\;\beta_{N-2}\rho_{N-1}^{2}\;-\;\beta_{0}\rho_{0}^{2},
  \ \
  D:=\sum_{j=0}^{N-2}\alpha_{j}(\rho_{j+1}-\rho_{j})^{2},
  \label{eq:quadform_expanded}
\end{equation}
whose middle sum is bounded by
$[\tfrac12\max_{x}(-f'(x))+\Order(\Delta x^{2})]\|\boldsymbol{\rho}\|^{2}$
through $(f_{j+1}-f_{j-1})/(4\Delta x)=\tfrac12f'(x_{j})+\Order(\Delta x^{2})$,
which is the interior estimate of \cref{eq:abscissa}; the last two terms
are the wall contributions, positive exactly when the drift points outward,
which is \cref{eq:Fbdry}.

To absorb them, note that for any $1\le K\le N$,
\begin{equation}
  \rho_{0}^{2}\;\le\;\frac{2}{K}\,\|\boldsymbol{\rho}\|^{2}
    \;+\;K\!\!\sum_{j=0}^{K-2}\!(\rho_{j+1}-\rho_{j})^{2},
  \label{eq:trace}
\end{equation}
with the mirror bound for $\rho_{N-1}^{2}$: from
$\rho_{0}=\rho_{j}-\sum_{i<j}(\rho_{i+1}-\rho_{i})$ and Cauchy--Schwarz,
$\rho_{0}^{2}\le2\rho_{j}^{2}+2j\sum_{i\le K-2}(\rho_{i+1}-\rho_{i})^{2}$,
and averaging over $j<K$ with $K^{-1}\sum_{j<K}2j=K-1$ gives
\cref{eq:trace}.  The second hypothesis of \cref{eq:abscissa_hyp} and
$|f_{j}-f_{j+1}|\le\Lambda\Delta x$ give
$\alpha_{j}\ge\sigma/\Delta x^{2}-\Lambda/4\ge\sigma/(2\Delta x^{2})$, hence
$D\ge\frac{\sigma}{2\Delta x^{2}}\sum_{j}(\rho_{j+1}-\rho_{j})^{2}$, while
$-\beta_{0}\le[-f_{0}]_{+}/(2\Delta x)+\Lambda/4\le M$ and
$\beta_{N-2}\le M$.  Applying \cref{eq:trace} at both walls and enlarging
each Dirichlet sum to the full bond set,
\begin{equation}
  \beta_{N-2}\rho_{N-1}^{2}-\beta_{0}\rho_{0}^{2}
  \;\le\;\frac{4M}{K}\,\|\boldsymbol{\rho}\|^{2}
    \;+\;2MK\!\sum_{j=0}^{N-2}\!(\rho_{j+1}-\rho_{j})^{2}.
  \label{eq:wall_absorbed}
\end{equation}
Choosing $K:=\lfloor\sigma/(4M\Delta x^{2})\rfloor$, at least $1$ by the
third hypothesis, gives $2MK\le\sigma/(2\Delta x^{2})$, so the gradient term
is dominated by $D$ and cancels against it in
\cref{eq:quadform_expanded}; and $\lfloor y\rfloor\ge y/2$ for $y\ge1$ gives
$K\ge\sigma/(8M\Delta x^{2})$, whence $4M/K\le32M^{2}\Delta x^{2}/\sigma$.
With the interior estimate this bounds
$\boldsymbol{\rho}^{\!\top}\!H_{1}\boldsymbol{\rho}/\|\boldsymbol{\rho}\|^{2}$
by the right-hand side of \cref{eq:abscissa} for every real
$\boldsymbol{\rho}\neq0$, and the supremum yields \cref{eq:abscissa}.  Since
$M\Delta x=F_{\partial}/2+\Lambda\Delta x/4$ the wall term tends to
$8F_{\partial}^{2}/\sigma$; when $F_{\partial}=0$ it is
$2\Lambda^{2}\Delta x^{2}/\sigma$ and vanishes with the grid, leaving the
interior estimate alone.
\qed

\paragraph{Sharpness, and hypothesis (iii)}
The third hypothesis reads $\Delta x\lesssim\sigma/(2F_{\partial})$: the
boundary layer, of width set by $\sigma$ and the outward drift, must be
resolved.  It is at most four times stronger than the positivity condition,
since $F_{\partial}\le F$.  The bistable grid of \cref{sec:numerics}
satisfies it at $4M\Delta x^{2}/\sigma=0.37$; the logistic comparison run
does not at $n=6$ ($1.70$) and meets it from $n=7$ ($0.82$), so we report
its abscissa as a measurement rather than an instance of the bound.  The
bound is not sharp, and is not meant to be: its role is to establish that
$\lambda_{\max}(H_{1})$ does not track $1/\Delta x$.  Refining $n=5\to10$
bears that out, $0.71\to0.73$ (bistable) and $1.27\to2.18$ (logistic), with
$\alpha(A)=0$ throughout, and the split in \cref{eq:abscissa} predicts which
moves: the interior term is $2.87$ against $1.05$, the wall scale
$F_{\partial}^{2}/\sigma$ is $1.76$ against $12.96$, so the wall is
invisible in the first case and dominant in the second.  That is what
hypothesis (iii) is for --- the logistic layer, of width
$\sigma/F_{\partial}\approx0.03$, spans about one cell at $n=6$, and
resolving it and satisfying (iii) are the same requirement.

\section{Proof of \texorpdfstring{\cref{rem:multidim_pauli}}{Proposition 7.5}}
\label{app:pauli_proofs}

Under \cref{eq:A_decomp} an axis-$k$ bond at carry level $m$ factorises as
$\bigl(\bigotimes_{i\ne k}E^{(N_{i})}_{p_{i}p_{i}}\bigr)\otimes
|p\rangle\langle p|_{\mathrm{pre}}\otimes\Sigma_{m}^{\pm}$.  Both projector
factors are diagonal and expand in $\{I,Z\}$ by \cref{eq:prefix_proj},
giving the bipartite form with $x$-vector $V_{m}$ on register $k$; strings
of the same axis and carry level share it and commute by
\cref{eq:commute_criterion}, and distinct axes are disjoint, so the counts
add to $\sum_{k}(2n_{k}+1)$.  For the sparsity, the prefix diagonal
$\alpha_{j(\mathbf{p}_{\perp},p,m)}$ has total degree $\le L$ in the full
multi-index by the joint hypothesis of \cref{sec:setup}, and the grid is
uniform along every axis, so the multi-index coordinate is affine in the
bits of all $n_{\mathrm{tot}}-m$ non-suffix qubits and \cref{lem:affine}
applies there; its Walsh support therefore has weight $\le L$, which gives
the term count.
\qed

\subsection*{Why the grid must be uniform}

The affinity of \cref{lem:affine} is specific to a \emph{uniform} grid: on a
non-uniform or adaptive grid $x_{j}$ is no longer affine in the bits of $j$,
a degree-$L$ drift is no longer multilinear of degree $L$ in them, and the
Hamming-weight bound is lost.  Uniformity is therefore a standing hypothesis
of \cref{thm:pauli_count} and of everything downstream of it.

\section*{Data availability}
The classical simulations reported in \cref{sec:numerics},
including the Pauli decompositions, the product-synthesis check, and the
figures, are reproduced by the code archived at \url{https://github.com/k-chandramouli/qsim-nonlinear-dynamics}. 

\section*{Acknowledgments}
KC acknowledges fruitful discussions with Elin Ranjan Das.
This work is supported by the U.S.\ Department of Energy, Office of Science,
Advanced Scientific Computing Research, under contract number DE-SC0025384.
The authors acknowledge that AI tools, specifically Claude Opus 4.8 and 5.0, were used to help with writing code,
formatting the paper, and performing literature surveys. All theorems,
proofs, numerical results, and references were independently checked by the
authors, who assume responsibility for all content.

\bibliographystyle{siamplain}
\bibliography{references_siam}

\end{document}